\documentclass[a4paper,twocolumn,11pt,unpublished]{quantumarticle}
\pdfoutput=1

\usepackage[utf8]{inputenc}
\usepackage[T1]{fontenc}
\usepackage[english]{babel}
\usepackage{amsmath,amssymb,mathtools}
\usepackage{booktabs}
\usepackage{tabularx}
\usepackage{array}
\usepackage{xcolor}
\usepackage{enumitem}
\usepackage{graphicx}
\usepackage{placeins}
\usepackage{microtype}
\usepackage{amsthm}
\usepackage{ulem}
\usepackage{hyperref}
\hypersetup{
  colorlinks=true,
  linkcolor=blue!45!black,
  citecolor=blue!45!black,
  urlcolor=blue!55!black
}
\graphicspath{{./}}


\newcommand{\frontier}{\texttt{Frontier}}

\newcommand{\A}{\mathcal{A}}

\newcommand{\GF}{\mathbb{F}_2}
\DeclareMathOperator*{\argmax}{arg\,max}
\newtheorem{proposition}{Proposition}
\newtheorem{corollary}{Corollary}

\begin{document}

\title{Approximating optimal decoding of quantum LDPC codes with narrow frontiers}

\author[AL]{Anthony Leverrier}
\email{anthony.leverrier@inria.fr}
\author[RU]{R{\"u}diger Urbanke}
\email{rudiger.urbanke@epfl.ch}
\affiliation[AL]{Inria Paris, France}
\affiliation[RU]{EPFL, Lausanne, Switzerland}

\begin{abstract}
We introduce the \frontier{} decoder, a pruned dynamic-programming decoder for sparse quantum decoding problems.  \frontier{} processes error variables in a chosen order, merges prefixes with the same residual syndrome and logical label, and approximates logical-coset posterior masses by retaining only a narrow scored frontier.  Without pruning, the recursion is exact ordered inference with exponential complexity. 

In the code-capacity setting, the decoder reaches thresholds close to optimal for the surface code and the color code. In the circuit-level noise model, it achieves state-of-the-art performance with a very small average retained list size: less than 100 for the gross code $[[144,12,12]]$ at a physical error rate of $0.001$. When the list size is constant, the decoder has linear complexity, suggesting the possibility of low-latency implementations.
\end{abstract}

\maketitle

Quantum error correction requires decoders that infer a recovery operation from a noisy syndrome record.  
For two-dimensional codes such as the surface code, decoders typically exploit the strong geometric structure of the problem~\cite{Dennis2002,Edmonds1965,Fowler2012Surface,Delfosse2021UnionFind,Bravyi2014TensorNetwork,Higgott2022PyMatching,HiggottGidney2025SparseBlossom}.  For general quantum LDPC codes, the Tanner graph is sparse but need not be geometrically local~\cite{MacKay2004SparseGraphQuantum,TillichZemor2014HGP,Breuckmann2021QLDPC,PanteleevKalachev2022GoodQLDPC,LeverrierZemor2022QuantumTanner,Bravyi2024BB}.  This motivates decoders that do not rely on geometry but still exploit sparsity. Many approaches have been explored in the literature: message-passing decoders with postprocessing~\cite{Roffe2020BP_OSD,Hillmann2025LSD,DeMarti2024BPOTF}, guided decimation~\cite{Yao2024BPGD,Gong2024GDG}, branch or beam-style variants~\cite{DeMarti2024ClosedBranch,YeWeckerDelfosse2025BeamSearch}, ensembling~\cite{Muller2025RelayBP,Koutsioumpas2025AutDEC}, modifying the Tanner graph to make it amenable to message passing~\cite{Maan2026GARI}, and short-path search in large graphs~\cite{AghababaieBeni2025Tesseract,Ott2025DecisionTree}. Machine-learning decoders are also rapidly improving and achieve state-of-the-art performance when they can be trained~\cite{Bausch2024AlphaQubit,Blue2025BBTransformer,Gu2026Cascade}.

We describe the \frontier{} decoder, motivated by the logical maximum-likelihood rule, i.e., the decision rule that minimizes the logical error probability. Given the observed syndrome, this rule assigns to each logical label the total posterior mass of all compatible error configurations in that logical class, and chooses the label with largest mass.

\frontier{} is an ordered syndrome decoder that uses pruning. Fix an order of the fault variables. After the first \(t\) variables have been processed, this order separates the variables into a processed prefix and an unprocessed suffix; we call this separation the cut at time \(t\). A syndrome check is active at this cut if it touches at least one variable in the prefix and at least one variable in the suffix. The active boundary is the set of all such checks.

For an active check, the prefix has already contributed part of the final syndrome value, but the suffix may still change it. For each prefix, the decoder records the residual syndrome, i.e., the syndrome contribution that the suffix must still supply on the active boundary, together with the logical label accumulated by the prefix. This pair is the boundary state of the prefix. A boundary state retained by the decoder is called a survivor, and the current list of survivors after the first $t$ variables have been processed is the frontier \(\mathcal{F}_t\). If two prefixes have the same boundary state, then they have the same feasible suffixes. Their probability masses can therefore be merged by summing.

A check is completed when all variables touching it lie in the processed prefix. Its syndrome value can then no longer change. At that moment, the decoder compares the value produced by the prefix with the observed syndrome bit. Prefixes with the wrong value are discarded; for prefixes with the right value, the completed check no longer needs to be stored.

Each survivor is assigned a score \(S=P+ C\). Here \(P\) is the log of the total probability mass of all prefixes merged into that survivor. The term \(C\) is a heuristic suffix-compatibility score, estimating how likely the unprocessed variables are to supply the remaining residual syndrome on the active boundary. The decoder retains only survivors whose score is within \(\Delta\) of the best score, subject to a hard cap of \(K\) survivors.

After all variables have been processed, every syndrome check has been completed. The decoder groups the retained states by logical class, sums their masses within each class, and outputs the class with the largest retained mass. Without pruning, the same recursion is exact variable elimination along the chosen order. With pruning, it becomes a small-list approximation to logical-class posterior inference.

The main empirical result is that this short-list dynamic program, with the simple deadline ordering described below and a local suffix-compatibility score, can compete with much more specialized search decoders on several benchmark families, while using very small average numbers of retained states.

The main contributions are:
\begin{enumerate}[leftmargin=*]
\item We derive an ordered dynamic program for logical maximum-likelihood decoding. For a fixed variable order, the unpruned recursion exactly computes the posterior mass of each logical class by merging prefixes with the same active residual syndrome and logical label.

\item We turn this exact recursion into the practical \frontier{} decoder. The approximation uses deadline ordering, a local suffix-compatibility score, score-gap pruning, a hard cap \(K\), and final aggregation by logical class.

\item We test \frontier{} on several sparse quantum decoding problems: rotated surface codes and color codes in the code-capacity model, and rotated surface code and BB codes under circuit-level noise. The results show near-optimal threshold behavior in the code-capacity setting and competitive performance for circuit-level noise with small retained lists.

\item We use the retained terminal lists to diagnose decoder failures. We distinguish support loss, where pruning removes the correct logical class, from terminal ranking failure, where the correct class remains in the list but is not ranked first. We find that increasing the score gap \(\Delta\) reduces support loss until terminal ranking becomes the limiting step.
\end{enumerate}

The rest of the paper is organized as follows.  Section~\ref{sec:related} positions \frontier{} relative to related decoders.  Section~\ref{sec:model} states the decoding model, and Section~\ref{sec:decoder} gives the ordered-frontier recursion and presents the \frontier{} decoder.  Section~\ref{sec:design} discusses ordering and committee variants.  Section~\ref{sec:results} reports the code-capacity and circuit-level noise experiments. The final sections discuss potential improvements and conclusions.

\section{Related decoders}
\label{sec:related}

The closest relatives of \frontier{} are decoders that keep an explicit boundary state or a small explicit set of candidates.  Trellis and variable-elimination decoders process variables in an order and summarize the past by the boundary information needed by the unprocessed suffix~\cite{Forney1973Viterbi,Bahl1974BCJR,Kschischang2001FactorGraphs,Dechter1999BucketElimination}.  Quantum trellis decoders make this connection explicit for stabilizer and convolutional codes~\cite{OllivierTillich2006Trellises,PelchatPoulin2013DegenerateViterbi,SaboAloshiousBrown2024TrellisQudit}. The unpruned \frontier{} recursion belongs to this family: once an order is fixed, prefixes with the same active residual and logical label have the same feasible suffixes. Their probability masses can therefore be summed, so degenerate representatives are collapsed into one boundary state rather than kept as separate candidates.

Tensor-network decoders are close at the level of inference: they approximate
logical-coset partition functions by contracting a graphical model while
controlling the size of an intermediate boundary through a bond-dimension
cutoff~\cite{Bravyi2014TensorNetwork,FerrisPoulin2014TensorQEC,Chubb2021GeneralTNDecoding,PiveteauChubbRenes2024TNBeyond2D}.
In a tensor-network decoder, the numerical width is controlled by the bond
dimension \(\chi\). In \frontier{}, the corresponding empirical width is the
retained frontier size \(|\mathcal F_t|\): the cap \(K\) is only a hard maximum,
while \(\Delta\) makes the realized width adaptive from cut to cut. The boundary
itself is different. It is not a geometric contraction boundary in a chosen
lattice embedding, but the set of parity-check or detector rows crossing the
ordered column cut, with each survivor labelled by an active residual syndrome
and a logical value. Thus \frontier{} is geometry agnostic in the sense
that it applies directly to an ordered parity-check or DEM matrix.

Tesseract and decision-tree decoders are also close in spirit because they search over sparse candidate corrections while pruning a rapidly growing graph of possibilities~\cite{AghababaieBeni2025Tesseract,Ott2025DecisionTree}.  An important distinction is merging: \frontier{} merges all prefixes that have the same active boundary residual and logical label before pruning the merged frontier.  This makes it a pruned ordered dynamic program rather than a decoder targeting the most likely error.

\section{Decoding model}
\label{sec:model}

We write the decoding problem first as a binary linear syndrome problem.  A fault vector
\begin{equation}
    x\in\GF^n
\end{equation}
produces syndrome
\begin{equation}
    s = Hx,
    \qquad
    H\in\GF^{m\times n}.
\end{equation}
Here \(H\) denotes either a parity-check matrix, or a Detector Error Model (DEM) matrix, depending on the application.  In quantum decoding, a logical matrix
\begin{equation}
    L\in\GF^{k\times n}
\end{equation}
maps the same fault vector to a logical class
\begin{equation}
    \ell = Lx.
\end{equation}
Here $k$ is the dimension of the logical space of interest: it is the number of logical qubits when we consider only bit flips, for instance, or twice as large if we consider both bit and phase flips.
Throughout, \(n\) is the number of local factors, equivalently the number of columns of the binary parity-check or DEM matrix.
The probability mass of the logical class \(\ell\) given syndrome $s$  is
\begin{equation}
    Z_\ell(s)=
    \sum_{\substack{x:\ Hx=s\\ Lx=\ell}}
    \Pr(x).
    \label{eq:coset_partition}
\end{equation}
The posterior probability of logical class \(\ell\) is
\(Z_\ell(s)/\sum_{\ell'} Z_{\ell'}(s)\), so maximizing posterior probability is
equivalent to maximizing \(Z_\ell(s)\).
The same notation also covers finite local factors.  In general, variable \(j\) has a local alphabet \(\A_j\), prior probabilities \(p_j(x_j)\), a syndrome \(h_j(x_j)\in\GF^m\), and a logical value \(\ell_j(x_j)\in\GF^k\).  The binary model above is the special case \(\A_j=\GF\), \(h_j(1)=H_{\cdot j}\), \(h_j(0)=0\), and similarly for the logical values.  Given an observed syndrome \(s\), the decoder seeks the logical class \(\ell\) with the largest posterior mass among local outcomes \(x=(x_1,\ldots,x_n)\in\A_1\times\cdots\times\A_n\) satisfying
\begin{equation}
    \bigoplus_{j=1}^n h_j(x_j)=s,
    \qquad
    \bigoplus_{j=1}^n \ell_j(x_j)=\ell .
\end{equation}
This is the factor-model version of the same \(H,s,L,\ell\) problem.
All notions used below are therefore defined directly from the ordered columns \((h_j,\ell_j)\), the target syndrome, and the priors; no geometric embedding of the checks or variables is required.

In the code-capacity setting, the local alphabet is \(\A_j=\{0,1\}\) when decoding one type of error in a binary CSS code. For depolarizing noise, we can instead use one Pauli factor per qubit, with alphabet \(\A_j=\{I,X,Y,Z\}\) and priors \((1-p,p/3,p/3,p/3)\). This keeps the \(X/Z\) correlation induced by \(Y\) errors.

For circuit-level noise, the variables can be binary detector-error mechanisms or higher-alphabet circuit-location factors. A two-qubit depolarizing fault, for example, can be represented as a single local factor with outcomes \(II\) and the 15 non-identity two-qubit Paulis, with respective priors \(1-p\) and \(p/15\).

The rows of \(H\) correspond to code checks or to circuit-level detectors. When the distinction is not important, we use the words check and detector interchangeably.

\section{The \frontier{} decoder}
\label{sec:decoder}

\subsection{Exact active-boundary states}

Fix an order \(\pi=(j_1,\ldots,j_n)\) of the columns or factors, and write
\(\pi_{\le t}=(j_1,\ldots,j_t)\).  After the first \(t\) entries in this order
have been processed, let \(\Gamma_t\) be the ordered set of active checks:
these are the checks that touch at least one processed factor and at least one
unprocessed factor. See Fig.~\ref{fig:frontier_schematic} for an illustration.

A check is completed once all factors touching it have been processed. From
that point on, its syndrome value can no longer change. Therefore, any prefix
that gives the wrong value for a completed check is discarded. Checks that touch
only unprocessed factors do not yet depend on the prefix and do not need to be
stored.

For a binary partial path
\(x_{\le t}=(x_{j_1},\ldots,x_{j_t})\in\GF^t\) on the processed variables, the
active residual syndrome is
\begin{equation}
    q_t
    =
    s_{\Gamma_t}
    \oplus
    H_{\Gamma_t,\pi_{\le t}}x_{\le t}
    \quad\in\GF^{|\Gamma_t|}.
    \label{eq:frontier_residual}
\end{equation}
For finite-alphabet factors, the same prefix notation means
\(x_{\le t}=(x_{j_1},\ldots,x_{j_t})\in
\A_{j_1}\times\cdots\times\A_{j_t}\), and the matrix product is replaced by
the sum mod 2 of local syndromes.

A boundary state is the pair
\begin{equation}
    \sigma_t=(q_t,\lambda_t),
\end{equation}
where
\begin{equation}
    \lambda_t
    =
    \bigoplus_{r=1}^t \ell_{j_r}(x_{j_r})
\end{equation}
is the logical label accumulated by the prefix. The frontier table stores a
prefix log mass \(P_t(\sigma_t)\) for each reachable boundary state. It may also
store one arbitrary representative prefix for each state, only to output a
correction at the end.

Two prefixes with the same active residual syndrome and the same accumulated
logical label have identical feasible suffix sets. They differ only in the
probability mass already accumulated by the prefix. They can therefore be merged
by summing their masses.

This merging is also one point at which quantum degeneracy can help the
computation. A decoder that stores individual representative errors can spread
the posterior mass of one logical class over many competing candidates. In
\frontier{}, prefixes that induce the same active residual syndrome and logical
label are collapsed before pruning. Their masses add to \(P_t(\sigma_t)\), and
the merged state occupies only one entry of the candidate frontier.

\begin{proposition}[Unpruned exactness]
Fix a finite factor model with local alphabets \(\A_j\), priors \(p_j(x_j)\),
syndrome vectors \(h_j(x_j)\in\GF^m\), and logical values
\(\ell_j(x_j)\in\GF^k\). For any variable order \(\pi\), the unpruned
ordered-frontier recursion with exact merging by active residual and logical
label computes the logical-class masses
\begin{equation}
  Z_\lambda(s)
  =
  \sum_{\substack{x:\ \oplus_j h_j(x_j)=s\\ \oplus_j \ell_j(x_j)=\lambda}}
    \prod_j p_j(x_j).
\end{equation}
\end{proposition}

\begin{proof}
After \(t\) variables, the active residual syndrome records exactly the part of
the target syndrome that must still be supplied by the unprocessed variables on
the active checks. Completed checks have fixed values; prefixes that give the
wrong value on such a check have already been discarded. Checks that touch only
unprocessed variables are independent of the prefix.

It follows that two prefixes with the same active residual syndrome and the
same accumulated logical label have the same feasible suffix set. They differ
only in their accumulated probability mass. Merging them by summing their
masses therefore preserves the exact prefix mass for every reachable boundary
state.

Induction over \(t\) shows that the unpruned frontier table stores the exact
prefix mass for each boundary state. After all variables have been processed,
all checks have been completed. Grouping the remaining terminal states by
logical label gives exactly the desired logical-class masses \(Z_\lambda(s)\).
\end{proof}

\begin{corollary}[Exact frontier width]
If \(w_t=|\Gamma_t|\) active binary syndrome constraints are present after step
\(t\), and \(k\) logical bits are tracked, then the exact merged frontier
contains at most
\begin{equation}
    2^{w_t+k}
\end{equation}
boundary states at step \(t\).
\end{corollary}

For circuit-level noise, this bound quickly becomes astronomically large.  Under the deadline order used in our experiments, we find \(\max_t w_t=134\) for the gross code $[[144,12,12]]$ on both memory-\(X\) and memory-\(Z\), and \(k=12\), giving an unpruned bound of \(2^{146}\) boundary labels.  The rotated-surface circuit-level DEM reaches the same exponent at \(d=11\), where \(\max_t w_t=145\) and \(k=1\).  Thus exact unpruned frontier contraction is not a viable decoder for these instances.

\begin{figure*}[!tp]
\centering
\includegraphics[width=0.94\textwidth]{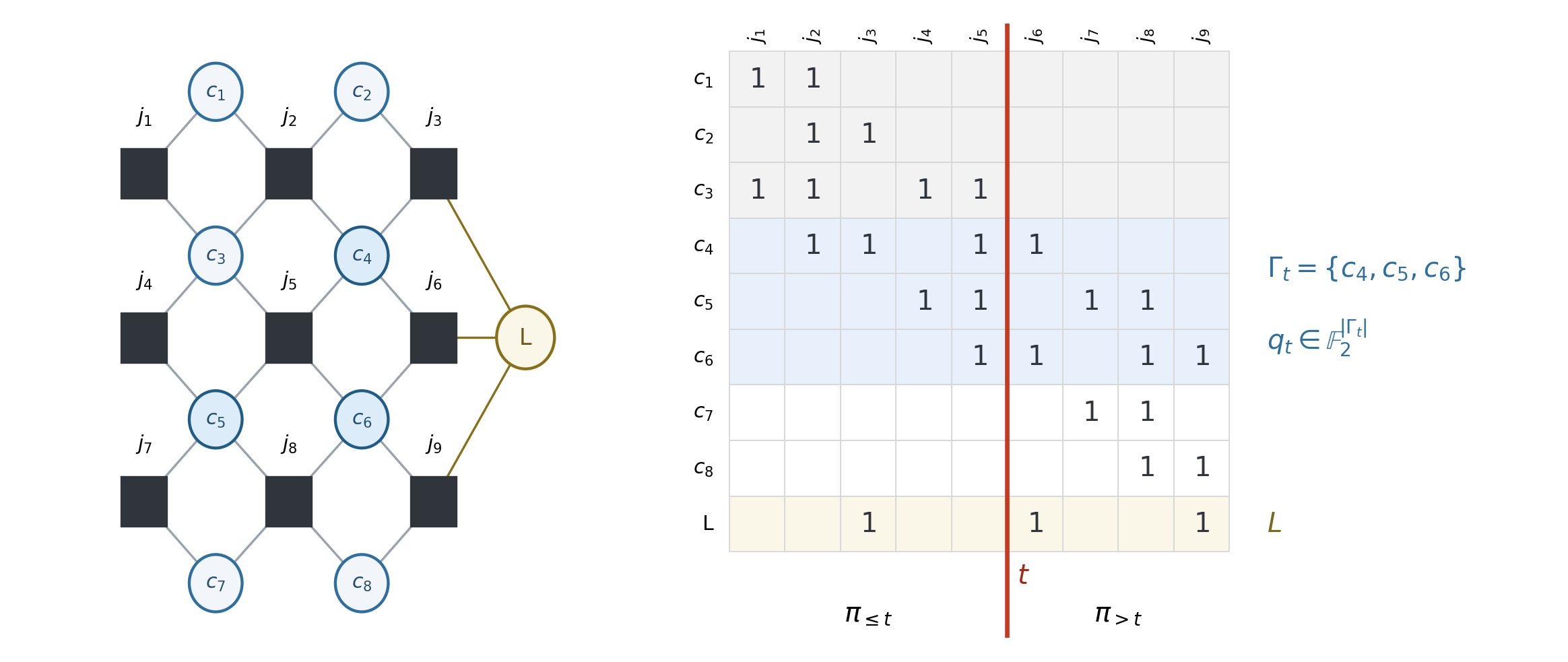}
\caption{Standard surface-code factor model and its ordered-matrix representation.  Left: one CSS side of a small planar surface-code patch, with local factors \(x_{j_i}\), bulk checks of degree four, boundary checks of degree two, and one logical-label row \(L\).  Right: the same incidence pattern written as the ordered matrix \(H\), with columns in the order \(\pi=(j_1,\ldots,j_n)\).  After step \(t\), the processed columns are \(\pi_{\le t}\), the unprocessed columns are \(\pi_{>t}\), and the active check set \(\Gamma_t\) is exactly the set of blue rows, namely the rows touching both sides of the cut.  The active residual \(q_t=s_{\Gamma_t}\oplus H_{\Gamma_t,\pi_{\le t}}x_{\le t}\in\mathbb{F}_2^{|\Gamma_t|}\) is therefore a vector with one residual syndrome bit for each blue row; the decoder stores \(q_t\) together with the accumulated logical label \(\lambda_t\).}
\label{fig:frontier_schematic}
\end{figure*}

\subsection{State transition}

Suppose that the next factor is \(j_t\). Each retained survivor from the current frontier \(\mathcal F_{t-1}\) branches over all possible values \(x\in\A_{j_t}\). For each child, the decoder adds \(\log p_{j_t}(x)\) to the prefix log mass, XORs
the syndrome \(h_{j_t}(x)\) into the residual syndrome, and XORs the
logical value \(\ell_{j_t}(x)\) into the accumulated logical label.

The active residual syndrome is then restricted to the new active boundary
\(\Gamma_t\). If a check has just become completed, its value is now fixed. The
decoder compares this value with the observed syndrome bit. If the two disagree,
the child is discarded. If they agree, the completed check no longer needs to be
stored.

After these updates, children with identical active residual syndrome and
logical label are merged by summing their masses. The resulting merged table is
the candidate frontier before pruning.

\subsection{Suffix-compatibility score and score-gap pruning}

The pruning step is applied after branching, discarding children that fail
completed checks, and merging duplicate boundary states. Let
\(\widetilde{\mathcal F}_t\) be the merged candidate frontier at time \(t\),
before pruning. Let \(\mathcal F_t\) be the retained frontier after pruning. We
call \(|\mathcal F_t|\) the retained list size at time \(t\).

The retained list size is the \frontier{} analogue of a tensor-network
bond dimension. The parameter \(K\) is only a memory cap; the realized profile
\((|\mathcal F_t|)_t\), and especially its average and peak over a decoding
instance, is the effective boundary width actually used by the decoder. This is
why our simulations below report retained-list statistics instead of the cap \(K\). 
We also discuss in Fig.~\ref{fig:gross_dem_avg_retained} the distinction between the peak frontier size $\max_t |\mathcal{F}_t|$ and its average $\frac{1}{n} \sum_t |\mathcal{F}_t|$.

For a candidate state \(\sigma\in\widetilde{\mathcal F}_t\), let \(q(\sigma)\)
be its active residual syndrome and let \(P_t(\sigma)\) be its prefix log mass.
Thus \(P_t(\sigma)\) is the log of the total probability mass of all prefixes
that have been merged into \(\sigma\).

We also assign to each active residual \(q\) a heuristic score \(C_t(q)\). We
call \(C_t\) the suffix-compatibility score. It is used only for pruning. Its
purpose is to favor residual syndromes that appear easier to complete using the
unprocessed variables.

Candidates are ranked by the pruning score
\begin{equation}
    S_t(\sigma)=P_t(\sigma)+\alpha C_t(q(\sigma)).
    \label{eq:score}
\end{equation}
The parameter \(\alpha\) controls how much the suffix-compatibility score
influences pruning relative to the prefix mass. In our experiments we use
\(\alpha=0.8\).

The best pruning score would use the log probability that the unprocessed
variables supply the remaining residual syndrome on all active checks. Computing
that probability exactly is too expensive. Instead, we use a row-wise
approximation. For each active check \(i\), we compute the probability that the
unprocessed variables supply the required residual bit on that check. We then
add the corresponding log probabilities over active checks. The row probabilities \(\rho_{i,t}\) are single-row marginals. The approximation
is to add their logarithms as if the active checks were independent. They are
not independent in general, because one future variable can touch several active
checks. Thus \(C_t\) is only a pruning heuristic.

We now define this score. Fix an active check \(i\in\Gamma_t\). Let \(U_t(i)\)
be the set of unprocessed variables that still touch check \(i\), and draw each
future local variable \(X_j\) independently from its prior \(p_j\). The future
tail contributes the row parity
\begin{equation}
    B_{i,t}
    =
    \bigoplus_{j\in U_t(i)} h_{j,i}(X_j).
\end{equation}
A state with active residual bit \((q_t)_i\) can be completed on this row only
if the future tail supplies parity \(B_{i,t}=(q_t)_i\).

For one future variable, define the parity moment
\begin{equation}
    \mu_{j,i}
    =
    \mathbb E\!\left[(-1)^{h_{j,i}(X_j)}\right]
    =
    \sum_{x_j\in\A_j}p_j(x_j)(-1)^{h_{j,i}(x_j)}.
\end{equation}
Equivalently, \(\mu_{j,i}=1-2\theta_{j,i}\), where
\begin{equation}
    \theta_{j,i}
    =
    \Pr\!\left[h_{j,i}(X_j)=1\right].
\end{equation}
Because parities multiply in the \((-1)^{(\cdot)}\) representation, the row
moment is
\begin{equation}
    \eta_{i,t}
    =
    \mathbb E\!\left[(-1)^{B_{i,t}}\right]
    =
    \prod_{j\in U_t(i)}\mu_{j,i}.
\end{equation}
The two probabilities of the binary parity \(B_{i,t}\) are then recovered from
this moment as
\begin{align}
    \rho_{i,t}(b)
    &=
    \Pr[B_{i,t}=b]\nonumber\\
    &=
    \frac{1}{2}\left(1+(-1)^b\eta_{i,t}\right),
    \qquad b\in\{0,1\}.
\end{align}
The suffix-compatibility score assigned to a residual vector \(q_t\) is the sum
of the row log-compatibilities:
\begin{equation}
    C_t(q_t)
    =
    \sum_{i\in\Gamma_t}\log \rho_{i,t}((q_t)_i).
    \label{eq:suffix_score}
\end{equation}

The retained frontier is
\begin{align}
    S_t^\star&=\max_{\sigma\in\widetilde{\mathcal F}_t}S_t(\sigma),\nonumber\\
    \mathcal F_t
    &=
    \operatorname{TopK}
    \left\{
      \sigma\in\widetilde{\mathcal F}_t:
      S_t(\sigma)\ge S_t^\star-\Delta
    \right\}.
    \label{eq:pruning}
\end{align}
Here \(\operatorname{TopK}\) means that, if more than \(K\) states pass the
score-gap test, only the \(K\) states with largest pruning score are retained.
Thus \(\Delta\) controls the score gap, while \(K\) is a hard memory cap. In
practice, we choose \(K\) much larger than the typical retained list size, and
the behaviour of the decoder is mostly determined by \(\Delta\).

A transition including branching, merging, scoring and pruning is illustrated
in Fig.~\ref{fig:algorithm}.

\begin{figure}[htp]
\begin{center}
\fbox{\begin{minipage}{0.92\linewidth}
\textbf{Algorithm 1: \frontier{} decoding for a fixed order.}
\smallskip

\textbf{Input:} factor model \((\A_j,p_j,h_j,\ell_j)\), syndrome \(s\), order
\(\pi\), cap \(K\), gap \(\Delta\), suffix weight \(\alpha\).\par
\textbf{Initialize:} \(\mathcal F_0=\{(q_0=0,\lambda=0,P=0,c=\varnothing)\}\),
where \(c\) is an optional representative prefix.\par
\textbf{For} \(t=1,\ldots,n\):
\begin{enumerate}[leftmargin=*]
    \item Let \(j=\pi(t)\).
    \item Generate children for every \(\sigma\in\mathcal F_{t-1}\) and every
    \(x\in\A_j\).
    \item Update the active residual syndrome, logical label, prefix log mass,
    and representative prefix \(c\).
    \item Discard children that give the wrong value on a newly completed check.
    \item Merge children with identical active residual syndrome and logical
    label by summing their masses, keeping one arbitrary representative prefix.
    \item Compute \(S_t(\sigma)\) for every state
    \(\sigma\in\widetilde{\mathcal F}_t\).
    \item Apply the score-gap rule and the cap \(K\) to obtain
    \(\mathcal F_t\).
\end{enumerate}
\textbf{Terminal:} aggregate the retained terminal states by logical class and
output the class with the largest retained mass, together with the corresponding
representative prefix if a correction is needed.
\end{minipage}}
\end{center}
\end{figure}

\begin{figure*}[htp]
\centering
\begin{minipage}[t]{0.45\textwidth}
\centering
\vspace{0pt}
\includegraphics[width=0.96\linewidth]{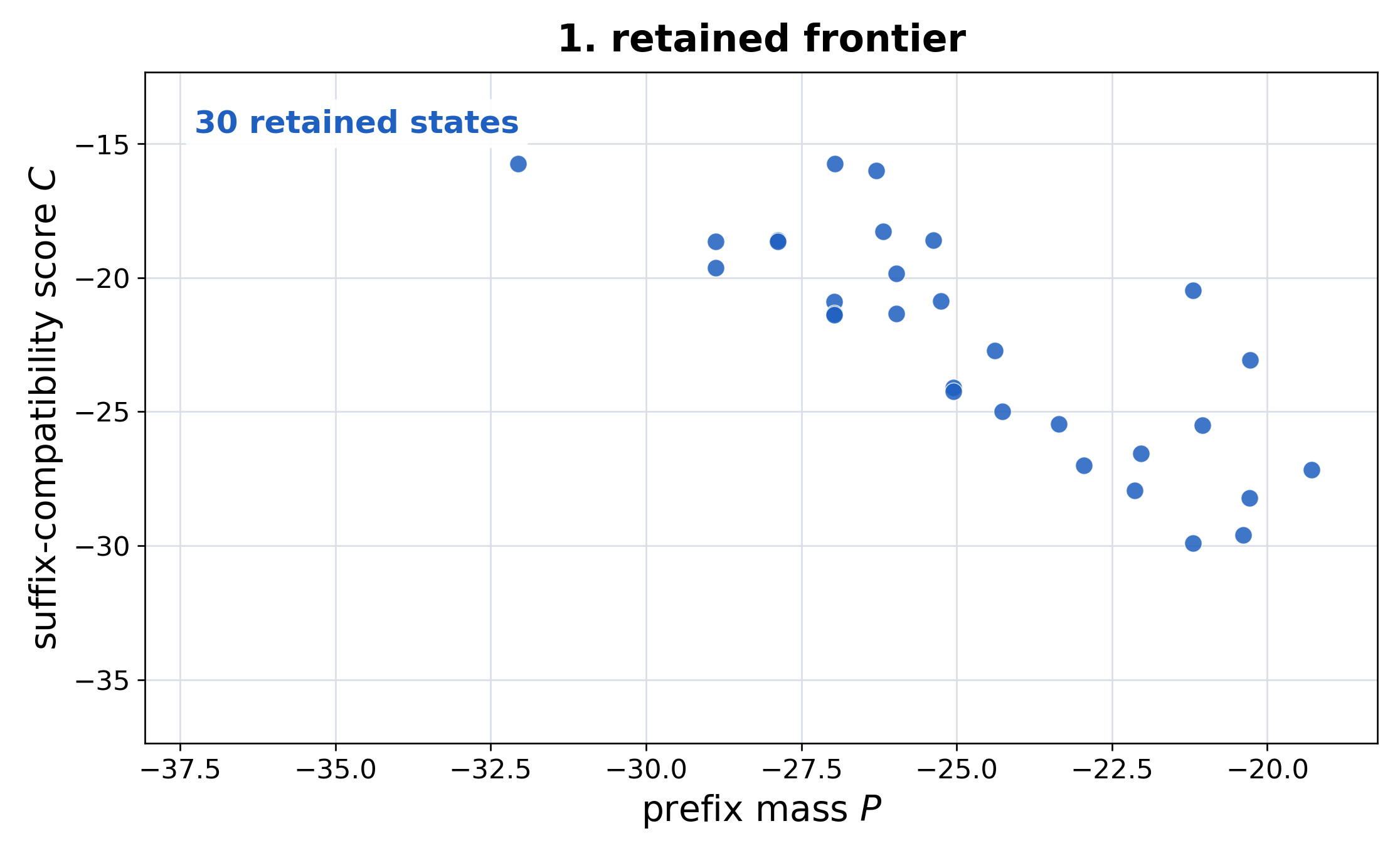}
\end{minipage}\hfill
\begin{minipage}[t]{0.45\textwidth}
\centering
\vspace{0pt}
\includegraphics[width=0.96\linewidth]{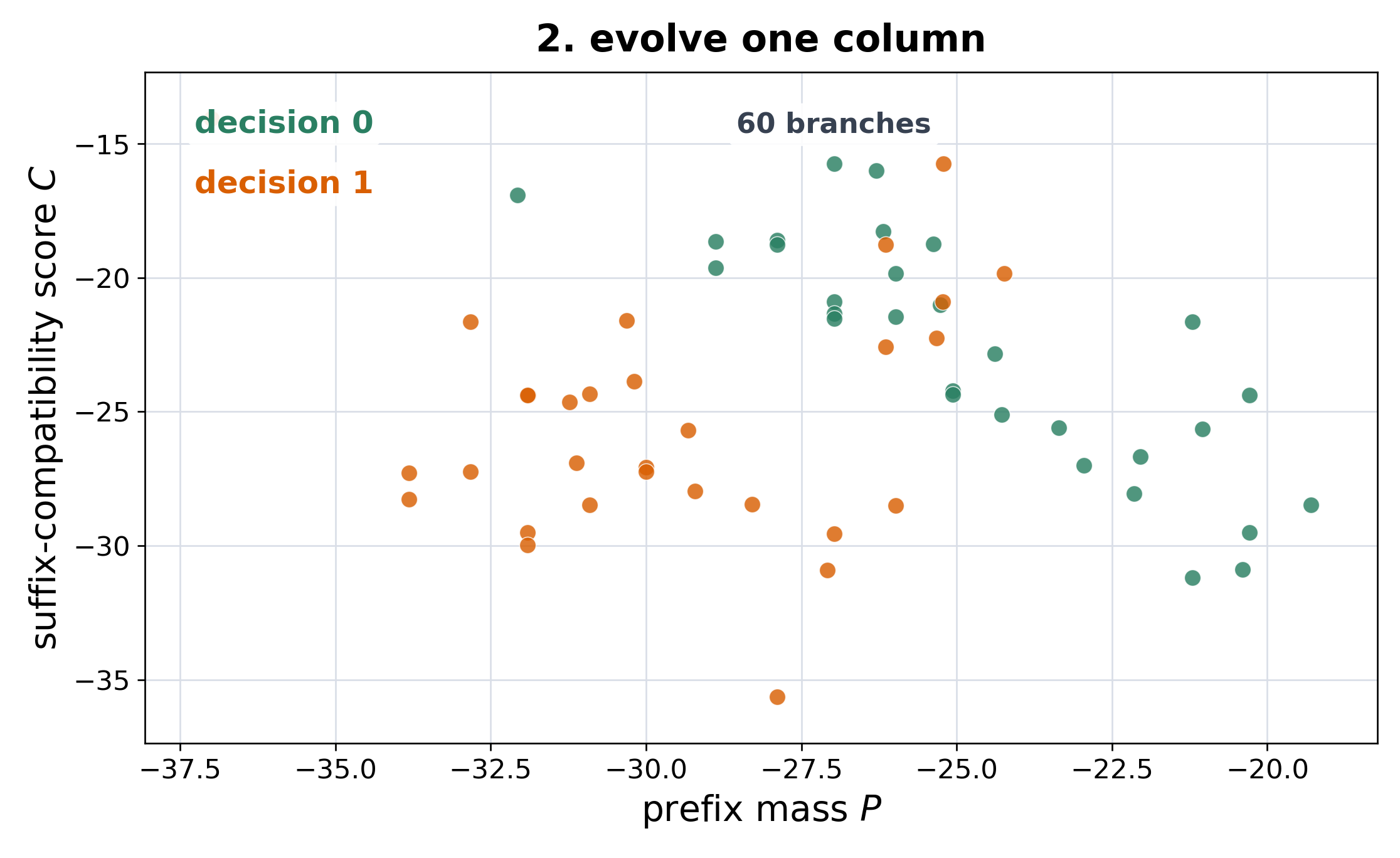}
\end{minipage}

\vspace{0.25em}

\begin{minipage}[t]{0.45\textwidth}
\centering
\vspace{0pt}
\includegraphics[width=0.96\linewidth]{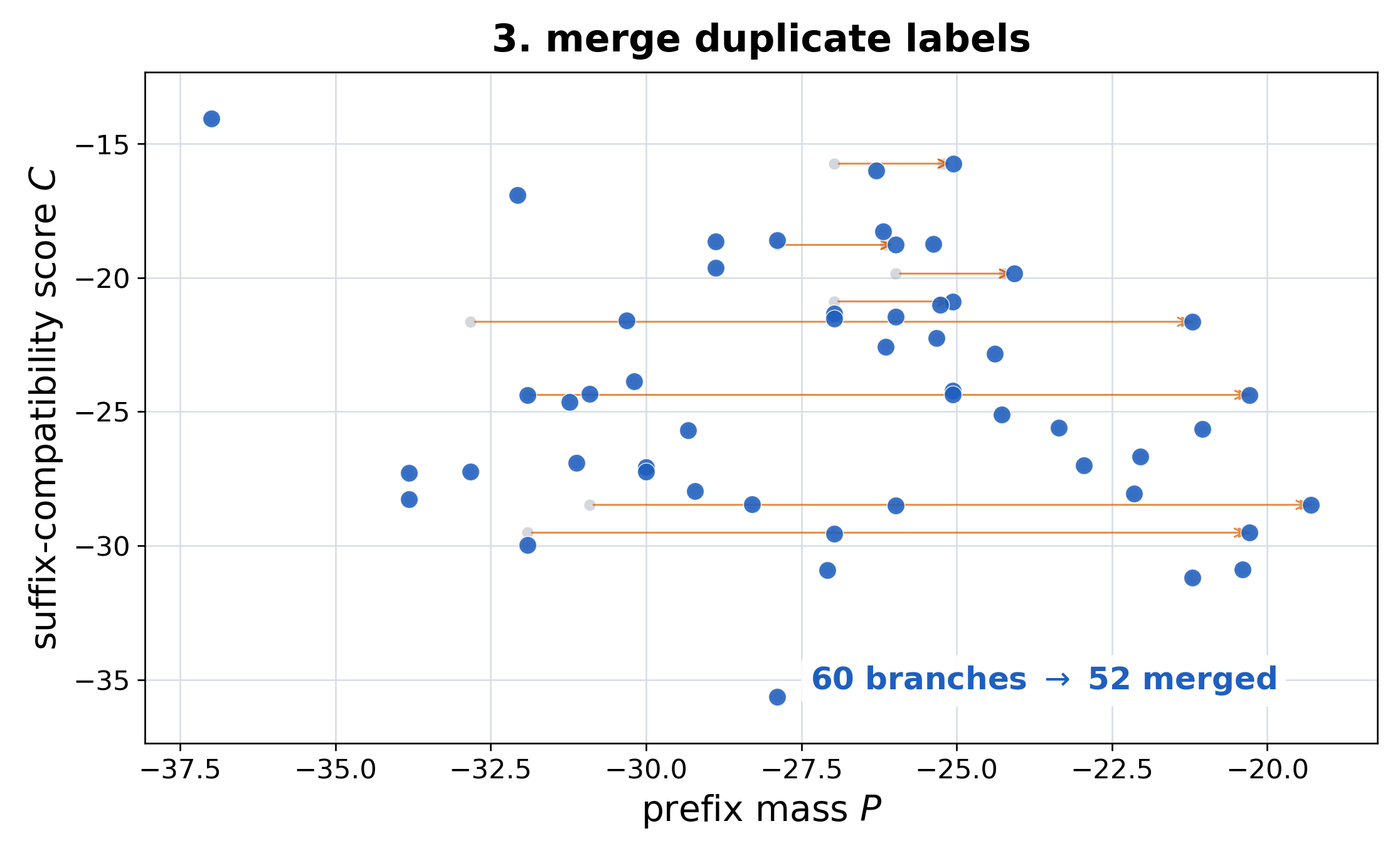}
\end{minipage}\hfill
\begin{minipage}[t]{0.45\textwidth}
\centering
\vspace{0pt}
\includegraphics[width=0.96\linewidth]{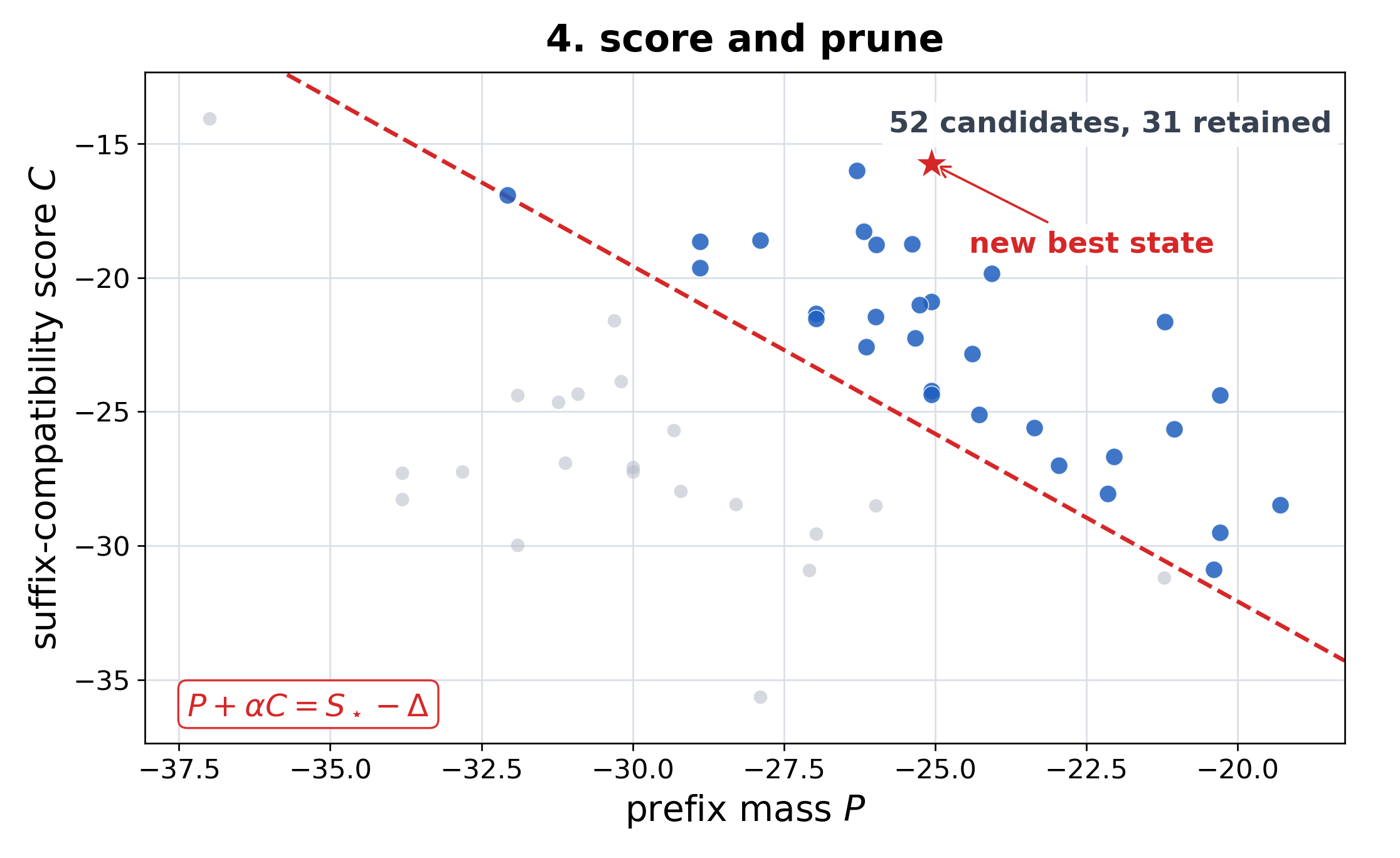}
\end{minipage}
\caption{One transition of the pruned ordered-frontier recursion on a BB code \( [[72,12,6]]\) under circuit-level noise. Retained states branch on the next binary detector variable, duplicate active-boundary labels are merged by summing their masses, and the merged candidates are scored by \(S=P+\alpha C\) before score-gap pruning.  In this boundary, the state count changes \(30\to60\to52\to31\).}
\label{fig:algorithm}
\end{figure*}

\subsection{Terminal logical decision}

After all variables have been processed, every check has been completed. The
decoder aggregates the retained terminal states by logical class. The retained
mass of logical class \(\lambda\) is
\begin{equation}
    \log \widehat Z^{(\mathcal F)}_\lambda(s)
    =
    \log\left(
    \sum_{\substack{\sigma\in\mathcal F_n\\ \lambda(\sigma)=\lambda}}
    \exp(P_n(\sigma))
    \right).
\end{equation}
The decoder outputs
\begin{equation}
    \widehat\lambda=\argmax_\lambda \widehat Z^{(\mathcal F)}_\lambda(s).
\end{equation}

The logical labels are not enumerated exhaustively in advance. Instead, each
retained survivor carries its accumulated logical label, and terminal
aggregation is performed only over the logical classes reached by the retained
frontier \(\mathcal F_n\).

\section{Design choices}
\label{sec:design}

\subsection{Column ordering}

We find that the decoder performance strongly depends on the choice of ordering $\pi$. Good orderings aim at keeping the active set of rows fairly small and at exposing contradictions early to avoid the need to carry a large frontier. A strategy that works well is a deadline order: variables are placed so that syndrome rows close early and the number of active constraints remains controlled.  The deadline of a row, denoted \(d_i\) below, is the position of its last incident variable under the current order.
The implementation used in the experiments starts from the natural order (e.g.~the temporal order of the fault locations in a quantum memory experiment) and then sorts variables by the earliest row deadline they can help close.  More explicitly, let \(R_j\subseteq\{1,\ldots,m\}\) be the detector/check rows touched by variable \(j\), and let \(r(j)\) be the position of \(j\) in the natural order.  For each touched row \(i\), define
\begin{equation}
    f_i=\min_{j:\ i\in R_j} r(j),
    \qquad
    d_i=\max_{j:\ i\in R_j} r(j).
\end{equation}
The forward deadline order sorts variables by the lexicographic key
\begin{equation}
    \kappa(j)=
    \left(
      \min_{i\in R_j} d_i,\,
      \max_{i\in R_j} d_i,\,
      \min_{i\in R_j} f_i,\,
      r(j),\,
      j
    \right),
    \label{eq:deadline_key}
\end{equation}
with the convention that a variable with \(R_j=\varnothing\) is assigned \(+\infty\) for the first three components, so that it is placed after variables that touch detector/check rows.  The first component prioritizes variables that can finish a row whose last possible incident variable is early.  The second and third components are deterministic tie breakers that avoid needlessly delaying columns whose support extends later, while preserving the original order as a final tie breaker.
This ordering is not a global optimization. In particular, applying the same rule in reverse coordinates gives a useful backward ordering; the committee variants below run both forward and backward deadline-ordered decoders.

\begin{figure}[htp]
\begin{center}
\fbox{\begin{minipage}{0.92\linewidth}
\textbf{Algorithm 2: Deadline ordering.}
\smallskip

\textbf{Input:} natural variable order, detector/check supports \(R_j\).\par
\textbf{For each} row \(i\): compute first touch \(f_i\) and last touch \(d_i\) in the natural order.\par
\textbf{For each} variable \(j\): compute the key \(\kappa(j)\) from Eq.~\eqref{eq:deadline_key}.\par
\textbf{Return:} variables sorted by increasing \(\kappa(j)\).  A backward deadline order is a separate permutation: replace \(r(j)\) by the reversed position \(n+1-r(j)\), compute the same key in those reverse coordinates, and sort by that key.
\end{minipage}}
\end{center}
\end{figure}

\subsection{Committee decoding}

Committee decoding is an optional extension after the single-order decoder has
been defined. The simplest committee runs two decoder instances, for example one
with the forward deadline order and one with the backward deadline order.

One selection rule is to compare the retained terminal logical-class masses
returned by the committee members and choose the class with the largest retained
mass. This is the rule we use in our simulations, with the forward/backward deadline order committee.

Another rule is disagreement escalation: if the committee members disagree,
rerun the decoder with larger \(K\) and \(\Delta\), up to a configured cap.

\section{Simulation results}
\label{sec:results}

Our primary accuracy metric is full logical frame error rate (FER). A frame
fails if the decoder outputs the wrong logical class, or if the final retained
list is empty. An empty final list means that the decoder did not retain any
error configuration with the observed syndrome. In our simulations, the
results use a committee of the forward and backward deadline orders: both orders are run on the same syndrome and the returned logical class with larger retained
terminal mass is selected.
All uncertainty intervals shown are Wilson 95\% confidence intervals computed
from frame counts. 

\subsection{Code-capacity threshold benchmarks}
\label{sec:surface}

The first code-capacity benchmark is the rotated surface code under depolarizing noise.  The correlated Pauli decoder uses one variable per qubit with alphabet \(\{I,X,Y,Z\}\) and prior \((1-p,p/3,p/3,p/3)\), so it keeps the \(X/Z\) correlation induced by \(Y\) errors. Figure~\ref{fig:surface_threshold} shows the performance of the \frontier{} decoder for \(\Delta=10\) and cap values $K$ increasing with the code size. The crossing window is close to the known depolarizing code-capacity maximum-likelihood threshold scale \(p\approx0.189\)~\cite{Bombin2012DepolarizationThreshold}.

\begin{figure*}[!tp]
\centering
\begin{minipage}[t]{0.49\textwidth}
\centering
\vspace{0pt}
\includegraphics[width=\linewidth]{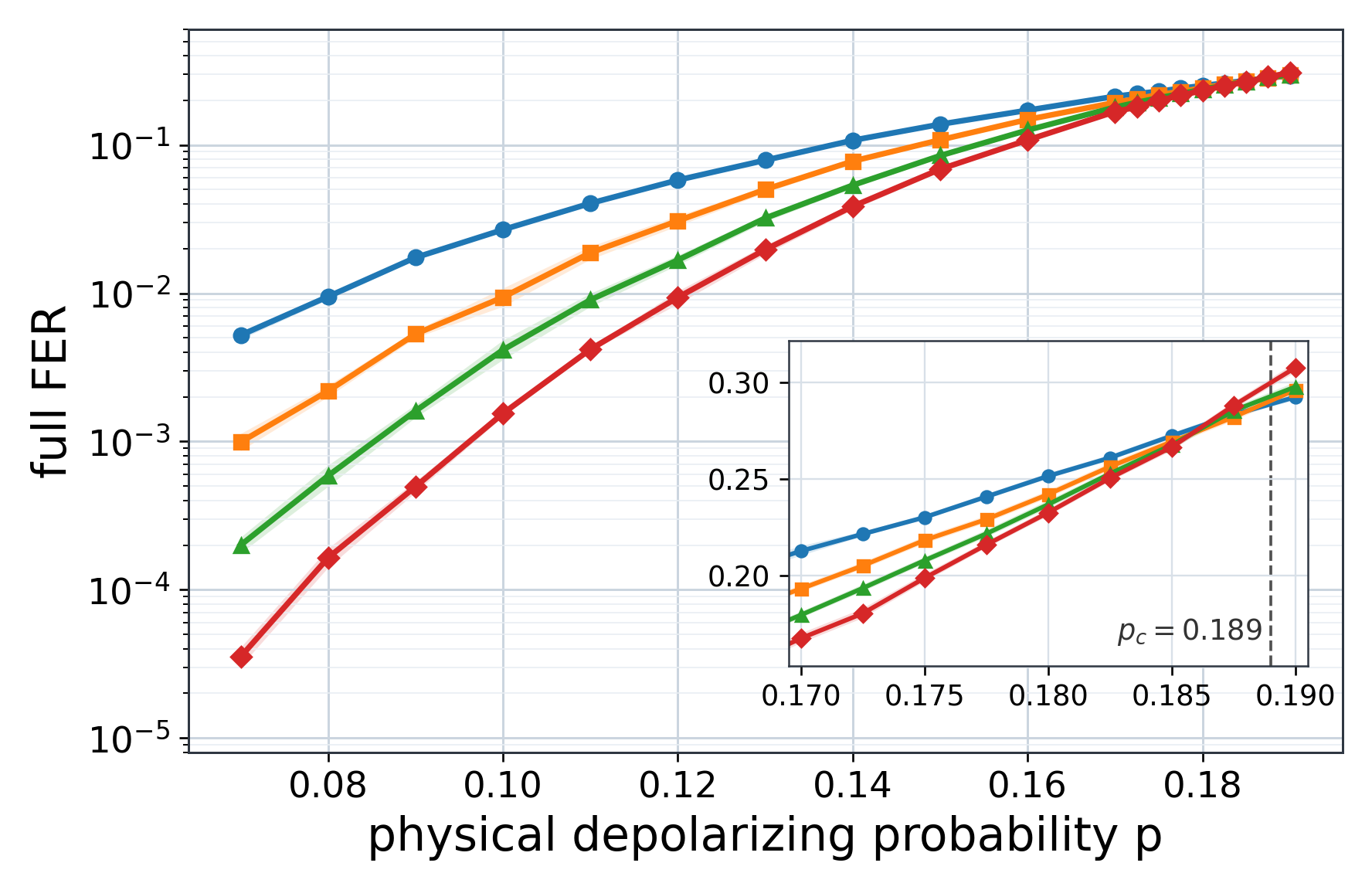}
\end{minipage}\hfill
\begin{minipage}[t]{0.49\textwidth}
\centering
\vspace{0pt}
\includegraphics[width=\linewidth]{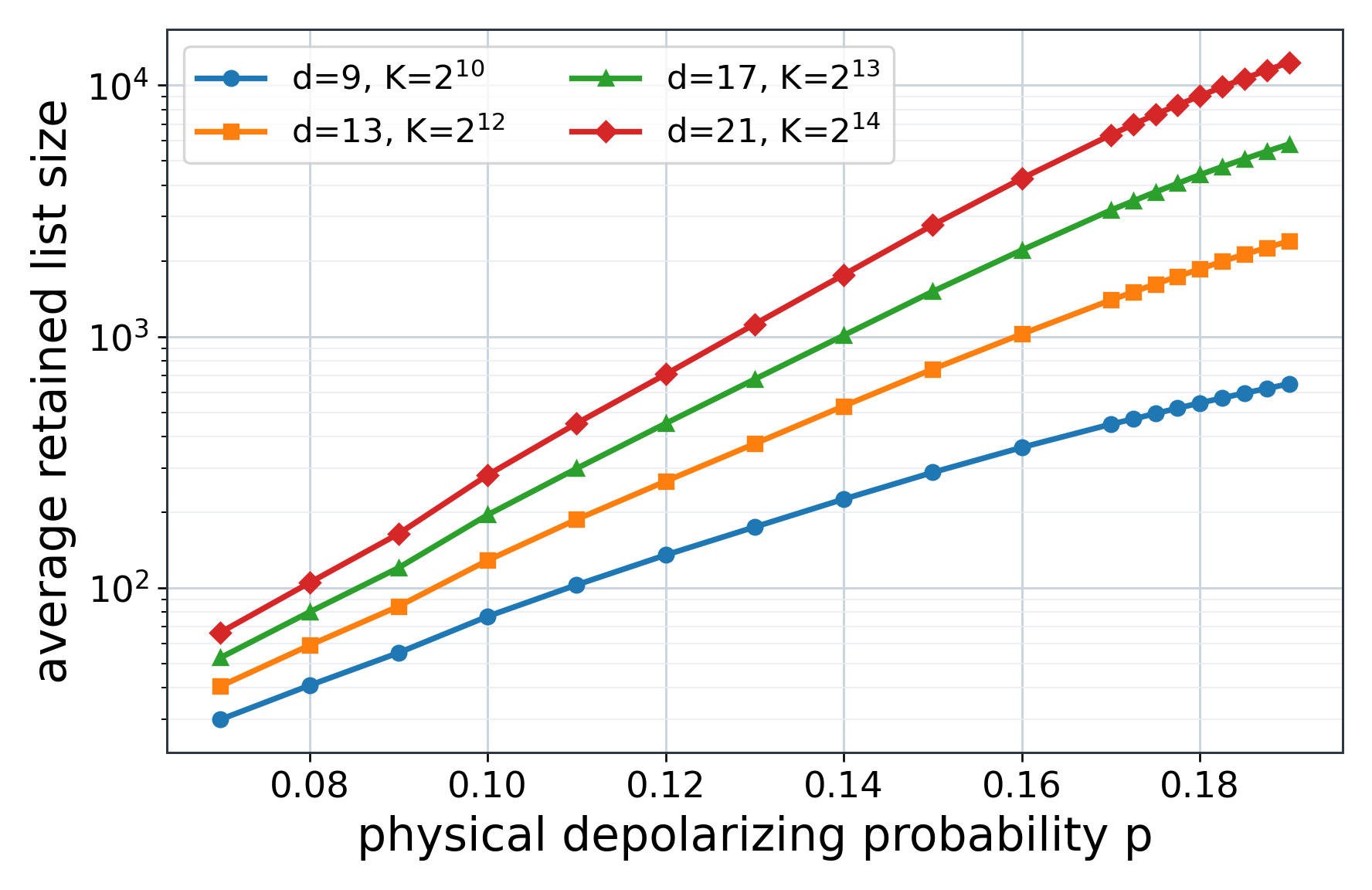}
\end{minipage}
\caption{Rotated-surface code-capacity depolarizing-noise behavior for correlated Pauli \frontier{} decoding.  Left: full FER versus physical depolarizing probability \(p\).  The dashed line marks \(p_c=0.189\), the optimal threshold obtained by computing the critical probability of a statistical-mechanics model~\cite{Bombin2012DepolarizationThreshold}.  Right: average retained list size $\frac{1}{n} \sum_t |\mathcal{F}_t|$ for the same setting.  Each qubit is represented by one \(I,X,Y,Z\) factor, all decoders use \(\Delta=10\), \(\alpha=0.8\), and a cap value $K$ that depends on the code distance.}
\label{fig:surface_threshold}
\end{figure*}

It is useful to compare the retained-list size in Fig.~\ref{fig:surface_threshold}
with the boundary size used in tensor-network decoders. 
A tensor-network decoder controls an MPS bond dimension $\chi$, and for a surface code of distance $d$, an MPS boundary contains $O(d\chi^2)$ scalar parameters. The bond dimension is typically a few dozen, for instance $\chi=48$ in \cite{Chubb2021GeneralTNDecoding}, albeit for larger distances. These numbers are comparable to the values of the average frontier size observed in Fig.~\ref{fig:surface_threshold} and suggest that both decoders find different efficient descriptions of the boundary of roughly the same complexity, either a
sparse list of residual-syndrome/logical labels with their accumulated masses for \frontier{}, or a dense low-rank boundary state for tensor-network decoding.

Our second code-capacity benchmark is the hexagonal color code under bit-flip
noise. Figure~\ref{fig:color_threshold} reports threshold curves and average
retained list sizes for distances \(d=9,13,17,21\). Again the results are consistent with the claim that the \frontier{} decoder approximates optimal decoding when the parameters $\Delta$ and $K$ are sufficiently large. The curves show that the
ordered-frontier approximation is not tied to surface-code geometry: the same
decoder is applied to a different parity-check matrix by changing the column
order and active-row boundary, not by designing a new two-dimensional
contraction geometry.

\begin{figure*}[!tp]
\centering
\begin{minipage}[t]{0.49\textwidth}
\vspace{0pt}
\centering
\includegraphics[width=\linewidth]{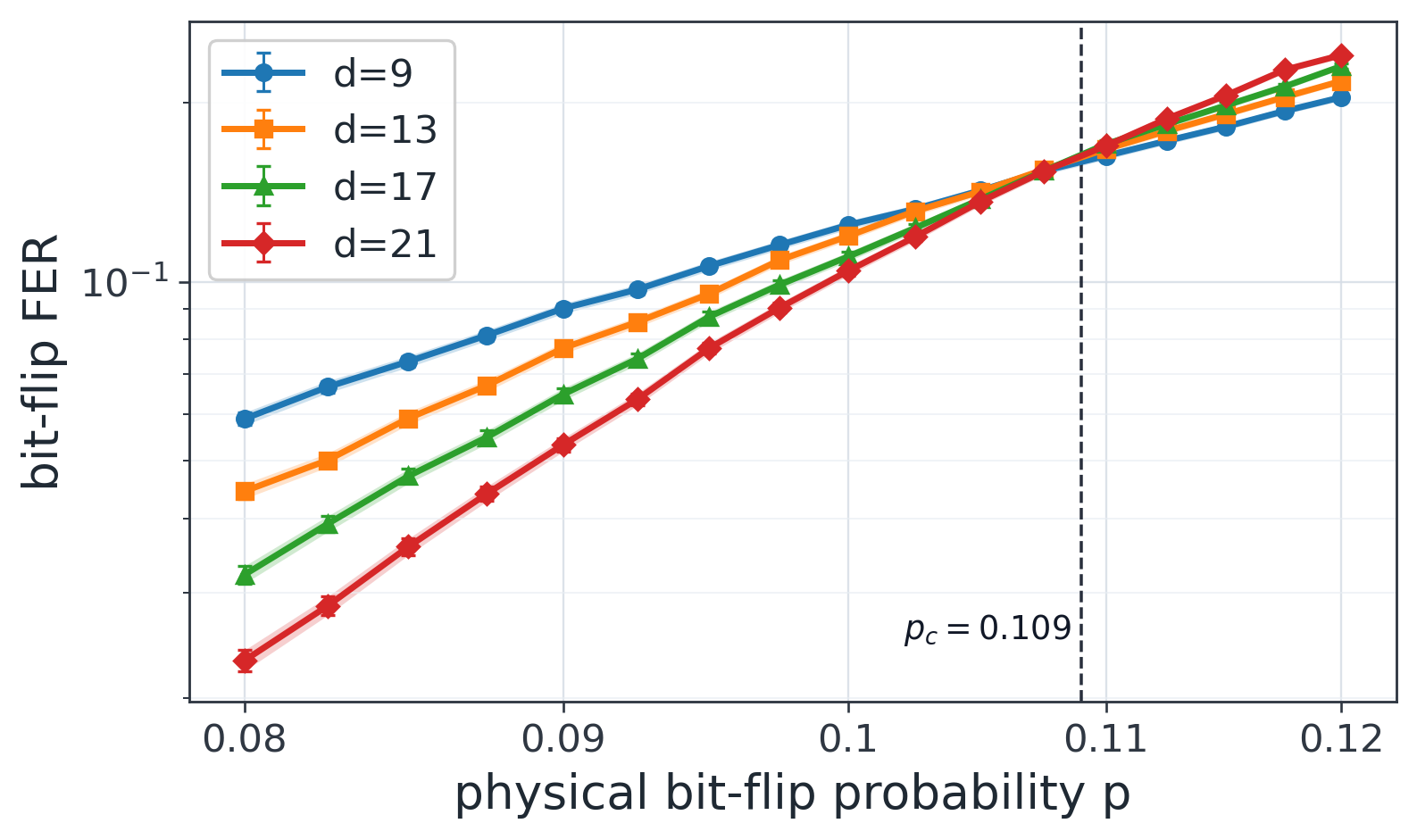}
\end{minipage}\hfill
\begin{minipage}[t]{0.49\textwidth}
\vspace{0pt}
\centering
\includegraphics[width=\linewidth]{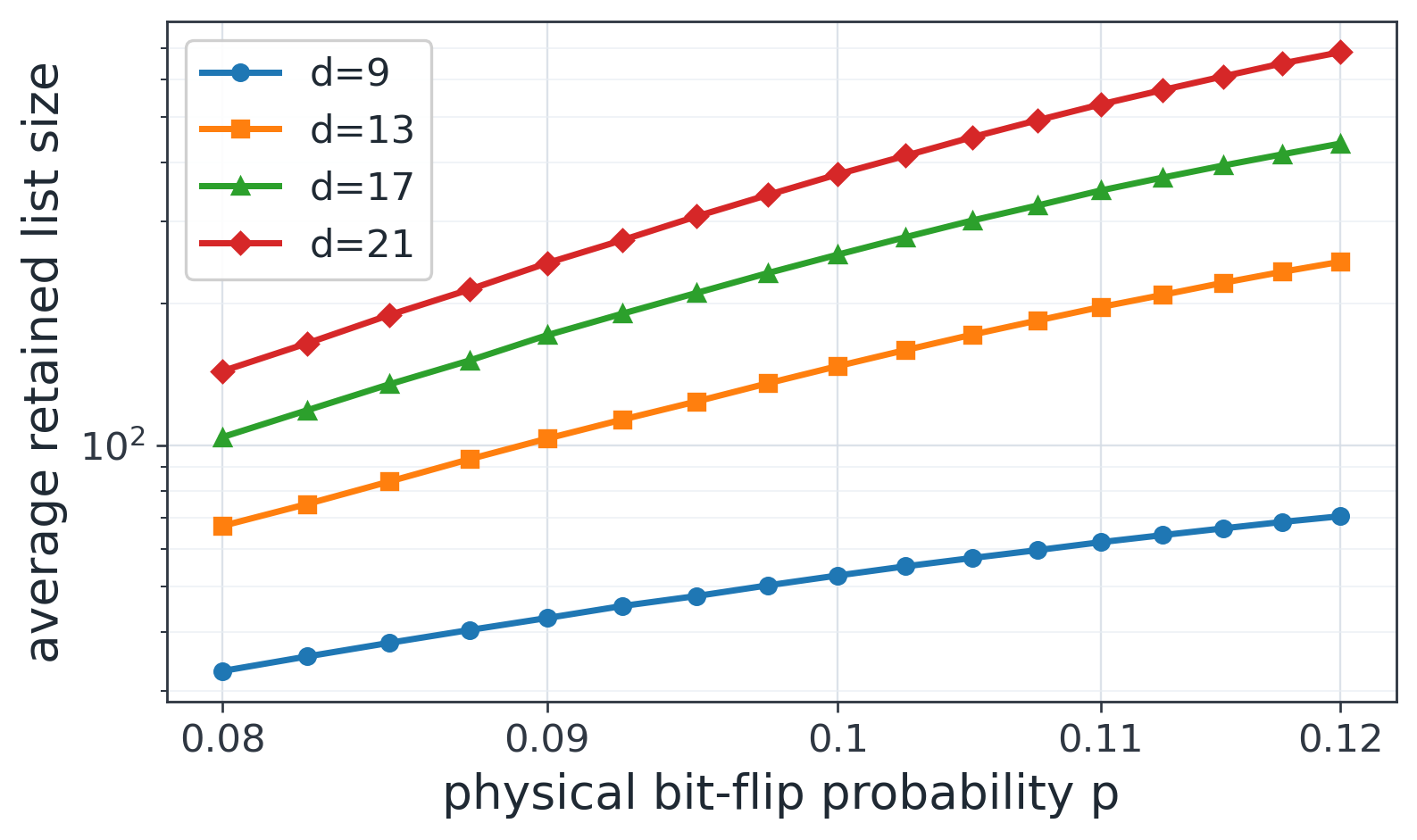}
\end{minipage}
\caption{Hexagonal color-code code-capacity bit-flip benchmark for
\frontier{} decoding. Left: FER versus physical bit-flip probability \(p\). Right: average retained list size, $\frac{1}{n} \sum_t |\mathcal{F}_t|$ for the same setting. The \frontier{} decoder uses \(\Delta=10\) and \(K=1024\) for all codes. The vertical dashed line the optimal threshold, obtained from a statistical-mechanics argument \(p_c\simeq0.109\)~\cite{KatzgraberBombinMartinDelgado2009ColorCodeThreshold}.}
\label{fig:color_threshold}
\end{figure*}

\subsection{Surface code under circuit-level noise}
\label{sec:surface-dem}

Figure~\ref{fig:surface_memory_z_dem_mwpm} compares \frontier{} decoding against correlated minimum-weight perfect matching~\cite{Higgott2022PyMatching,HiggottGidney2025SparseBlossom,PyMatchingSoftware} for the rotated surface code under depolarizing circuit-level noise at $p=0.002$. The plot displays the results for a memory-\(Z\) experiment rather than considering both $X$ and $Z$ decoding. 
Each experiment considers $d$ noisy syndrome extraction rounds and reports the corresponding FER per round.  

It is interesting to fit the per-round logical
error rate
\begin{equation}
    P_{\mathrm L}(d)
    \propto
     \Lambda^{-\left\lfloor (d+1)/2 \right\rfloor}
\end{equation}
in order to evaluate the error-suppression coefficient $\Lambda$. We find a value around $8.6$ for \frontier{} and $7.9$ for correlated MWPM, compatible with the figure reported in \cite{Gu2026Cascade} that also obtained $9$ for Tesseract (long beam) and $8.4$ for the neural decoder Cascade. It is likely that increasing $\Delta$ and $K$ to larger values for $d=11$ would improve the error-suppression coefficient of the \frontier{} decoder, but our current implementation is too slow to test this regime.

\begin{figure}[!tp]
\centering
\includegraphics[width=\linewidth]{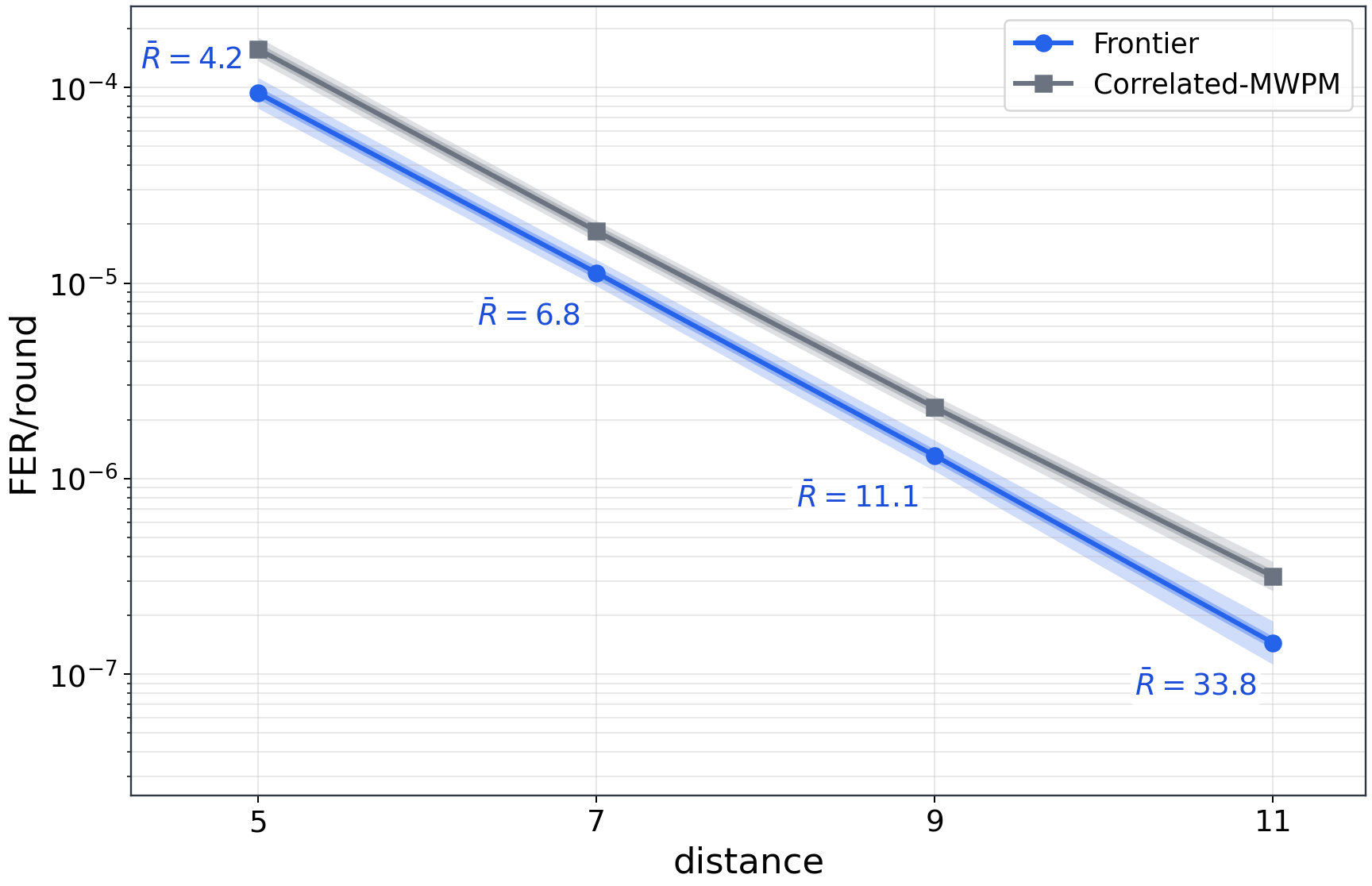}
\caption{Rotated-surface memory-\(Z\) memory experiment with $d$ syndrome extraction rounds and \(p=0.002\): logical FER per round versus code distance for \frontier{} and correlated minimum-weight perfect matching~\cite{Higgott2022PyMatching,HiggottGidney2025SparseBlossom,PyMatchingSoftware}. The decoder settings are $\Delta=10, K=1024$ for $d=5,7,9$ and $\Delta=12, K=2048$ for $d=11$.
Blue point labels report the estimated average retained frontier size \(\bar R=\frac{1}{n}\sum_{t=0}^{n-1}|\mathcal F_t|\), averaged over sampled shots and the two directional decoder passes.}
\label{fig:surface_memory_z_dem_mwpm}
\end{figure}

\subsection{BB codes under circuit-level noise}
\label{sec:bb-detector}

We ran simulations for two bivariate bicycle (BB) codes~\cite{Bravyi2024BB} with parameters $[[72,12,6]]$ and $[[144,12,12]]$. In both cases, we consider circuit-level depolarizing noise and respectively 6 and 12 rounds of syndrome extraction. We report full FER including both $X$- and $Z$-decoding, and do not normalize by round nor by logical qubit count. 
We compare the \frontier{} decoder to the beam search decoder \cite{YeWeckerDelfosse2025BeamSearch} in the configuration \texttt{beam64$\_$640iters} and Tesseract~\cite{AghababaieBeni2025Tesseract} with the short beam configuration. 
The DEM matrices have respective sizes \(H_X=H_Z=252\times2232\) and \(H_X=H_Z=936\times8784\), and the logical matrices have size  \(L_X=L_Z=12\times2232\) and \(L_X=L_Z=12\times8784\).

Figure~\ref{fig:bb_dem_circuit} reports both BB code results. The top row shows the $[[72,12,6]]$ code with the FER-vs-\(p\) comparison and the \(p=0.002\) list-size sweep. The bottom row shows the gross code $[[144,12,12]]$, where the list-size sweep is performed for $p=0.001$.
The main result of these simulations is that useful FER values are obtained with a very small retained frontier, at least in the regime of interest for fault-tolerant quantum computing, i.e.~for $p$ on the order of $10^{-3}$.

Figure~\ref{fig:gross_dem_avg_retained} compares the peak frontier size $\max_t |\mathcal{F}_t|$ and the average $\frac{1}{n} \sum_t |\mathcal{F}_t|$. The values are obtained by averaging these quantities over 5000 shots. For $p=0.001$, the figure shows that the peak frontier size during a shot may be approximately 5 times larger than the average size, confirming the usefulness of the parameter $\Delta$ to adapt the frontier size along the decoding. Without $\Delta$-adaptive pruning, one should instead keep a frontier with the peak size all along in order to obtain a similar accuracy.

\begin{figure*}[!tp]
\centering
\begin{minipage}[t]{0.49\textwidth}
\centering
\vspace{0pt}
\includegraphics[width=\linewidth]{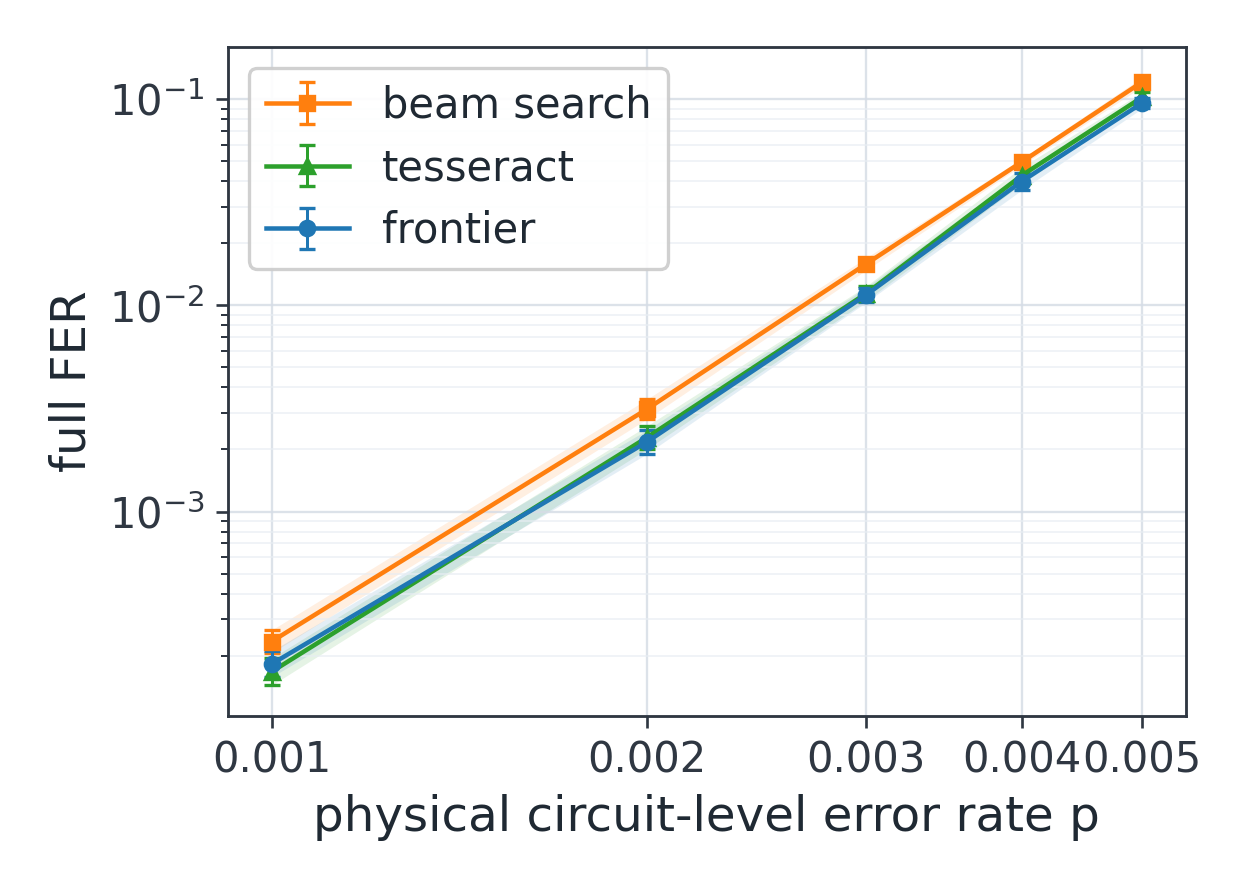}
\end{minipage}\hfill
\begin{minipage}[t]{0.49\textwidth}
\centering
\vspace{0pt}
\includegraphics[width=\linewidth]{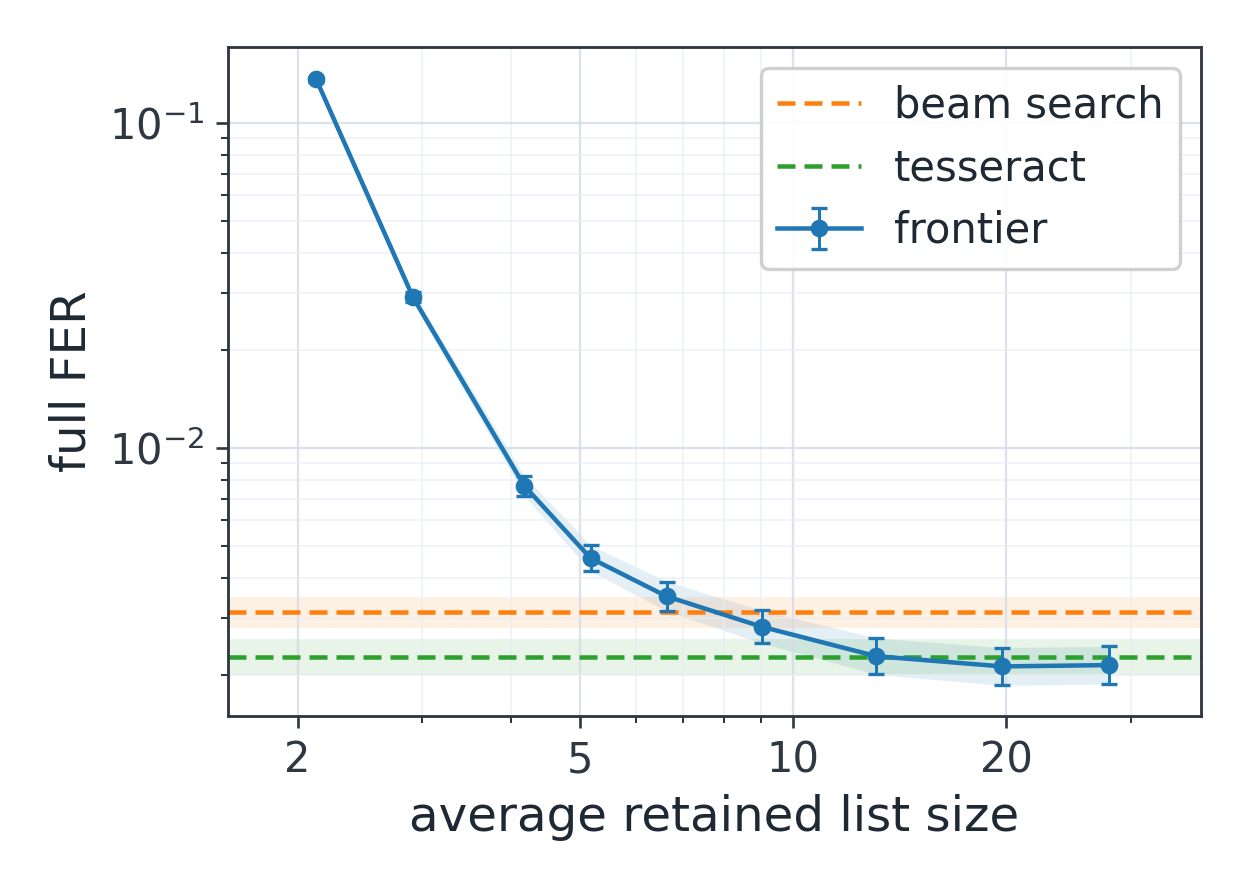}
\end{minipage}
\vspace{0.8em}

\begin{minipage}[t]{0.49\textwidth}
\centering
\vspace{0pt}
\includegraphics[width=\linewidth]{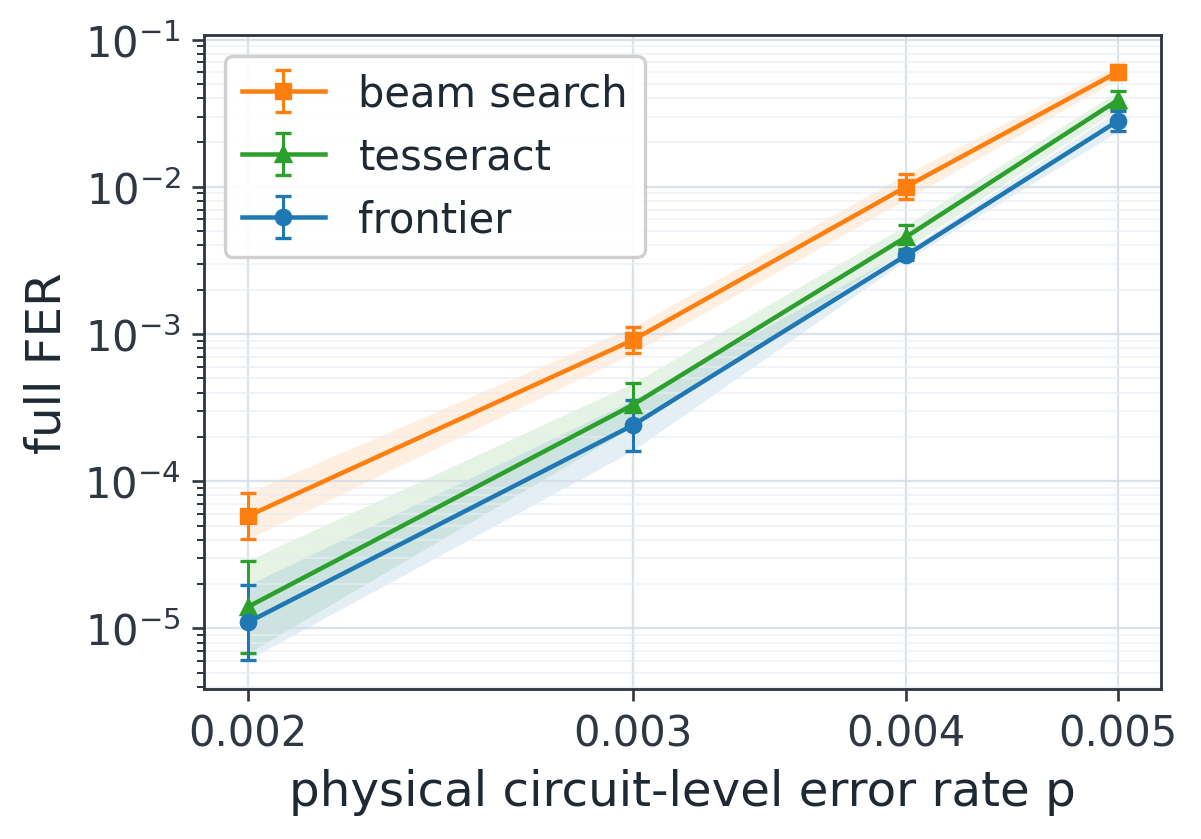}
\end{minipage}\hfill
\begin{minipage}[t]{0.49\textwidth}
\centering
\vspace{0pt}
\includegraphics[width=\linewidth]{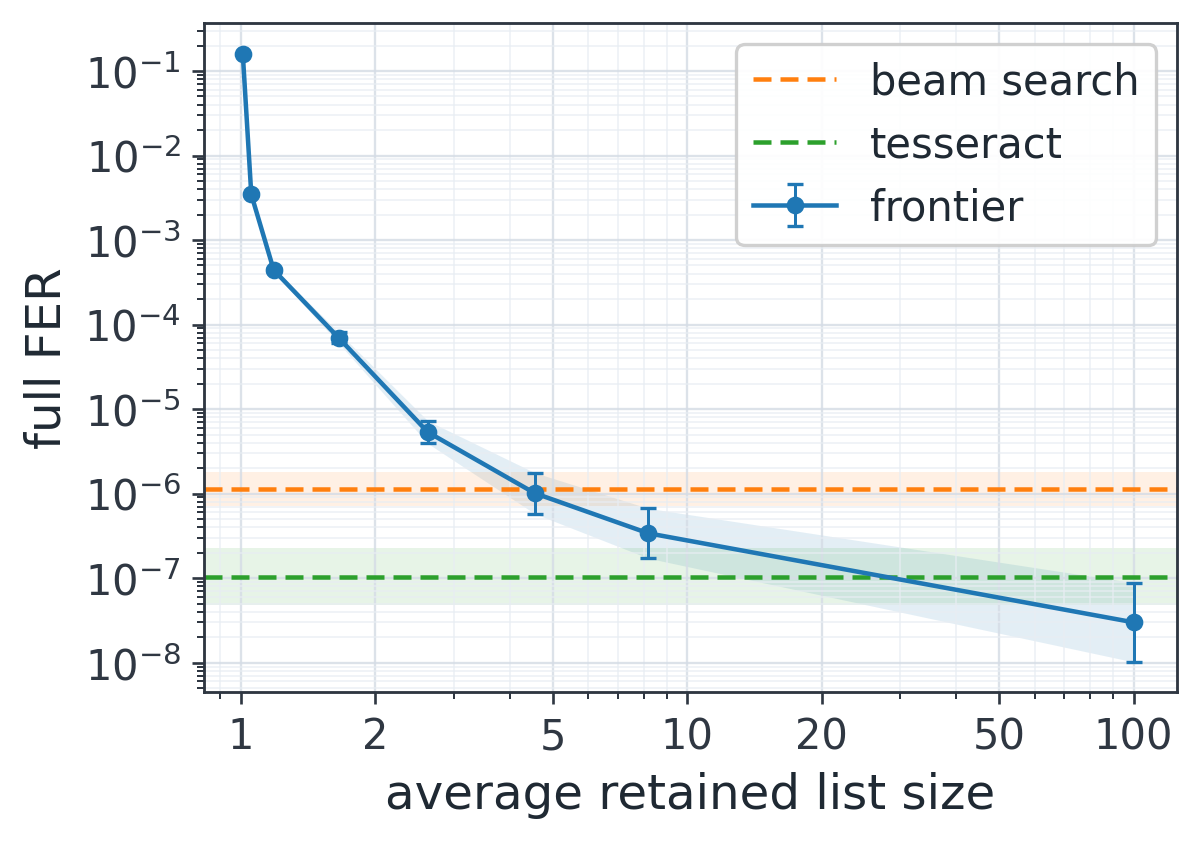}
\end{minipage}
\caption{BB circuit-level benchmarks. Top: BB code $[[72,12,6]]$, with full FER versus physical circuit-level error rate \(p\) on the left and full FER at \(p=0.002\) versus average retained list size for a \(K=4096\), \(\Delta=2,\ldots,10\) sweep on the right. Bottom: gross code $[[144,12,12]]$, with full FER versus \(p\) on the left and full FER at \(p=0.001\) versus average retained list size on the right; the cap is \(K=2^{14}\) for \(\Delta=20\) and \(K=2^{13}\) for \(2\leq\Delta\leq14\). Both rows compare against the beam search decoder~\cite{YeWeckerDelfosse2025BeamSearch} in the \texttt{beam64$\_$640iters} configuration and Tesseract~\cite{AghababaieBeni2025Tesseract} with the short beam configuration.}
\label{fig:bb_dem_circuit}
\end{figure*}

\begin{figure}[!tp]
\centering
\includegraphics[width=\linewidth]{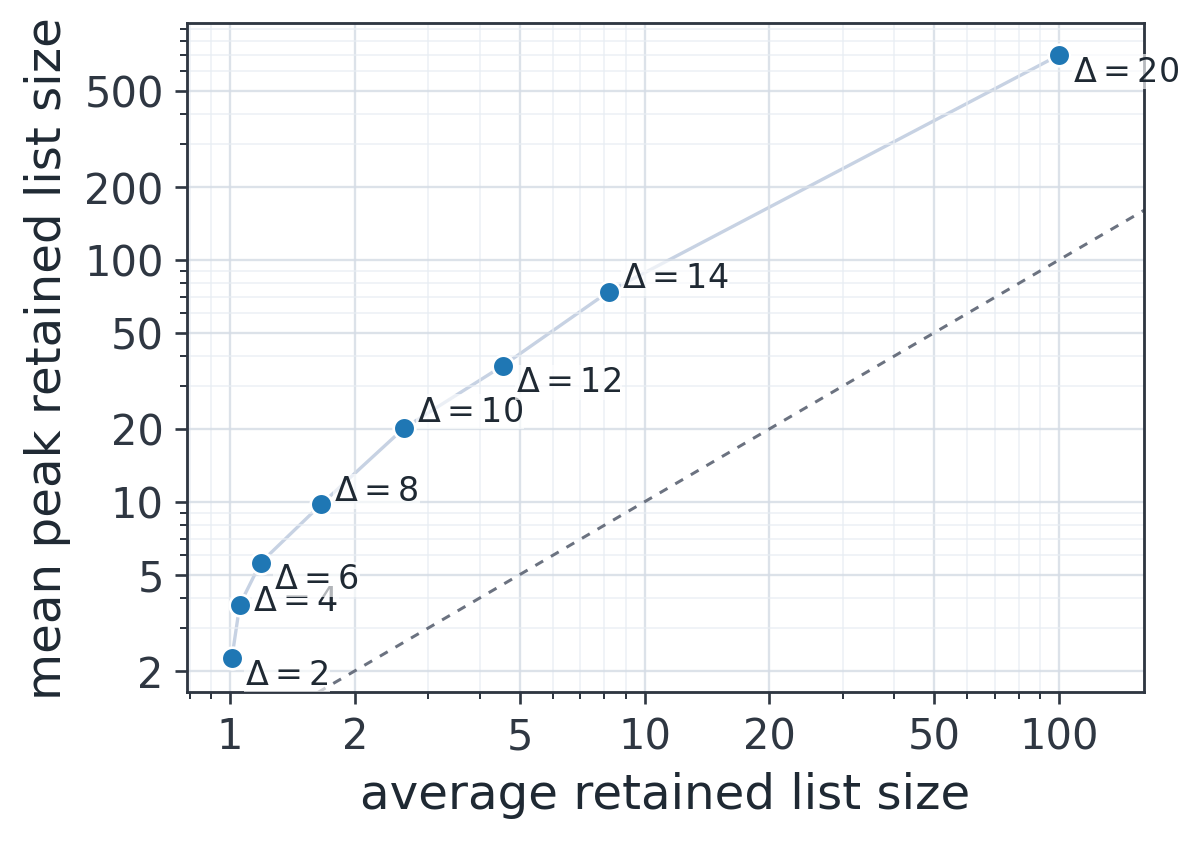}
\caption{Comparison between the peak frontier size and the average frontier size during a decoding of the gross code under circuit-level noise at $p=0.001$. The cap value is $K=2^{13}$ for all points, except for $\Delta=20$ where it is $2^{14}$.}
\label{fig:gross_dem_avg_retained}
\end{figure}

Figure~\ref{fig:transition_evals} shows the tail distribution of transition evaluations for the gross code under circuit-level noise with \(p=0.001\). Given that the complexity of the decoder is approximately linear in the number of transition evaluations, this figure shows the latency tails one can expect from a fast implementation of the \frontier{} decoder. The decoder settings here correspond to the strongest version of the decoder displayed in Fig.~\ref{fig:bb_dem_circuit}, namely $\Delta=20, K=2^{14}$. We see that the median number of transition evaluations is around $6 \times 10^6$, while the $99.9$th percentile is around $19 \times 10^6$, which is about three times larger. This suggests that a fast implementation of the decoder should not suffer from large latency tails.

\begin{figure}[!tp]
\centering
\includegraphics[width=\linewidth]{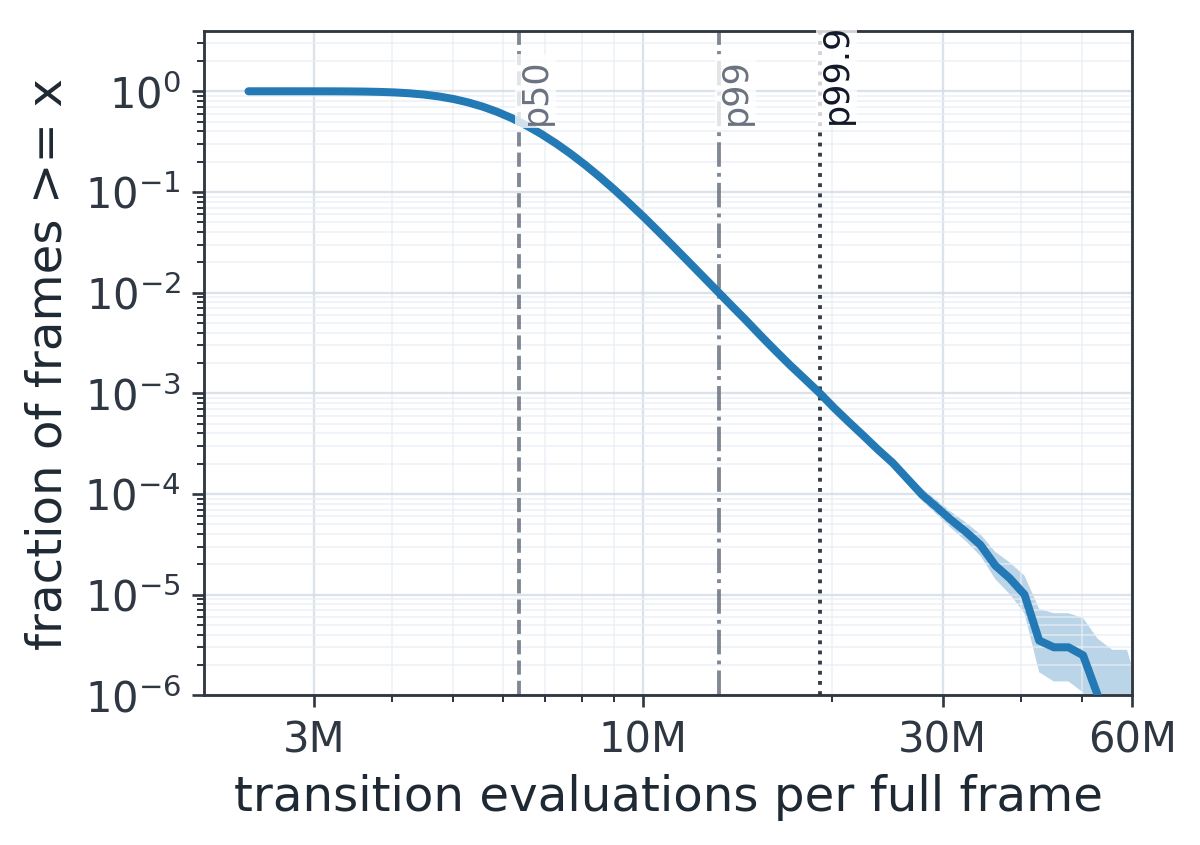}
\caption{Transition-evaluation tail curve for the gross code under depolarizing circuit-level noise $p=0.001$ with \(K=16384\), \(\Delta=20\).  A transition evaluation is one candidate extension before merging and pruning.}
\label{fig:transition_evals}
\end{figure}

\begin{figure}[!tp]
\centering
\includegraphics[width=\linewidth]{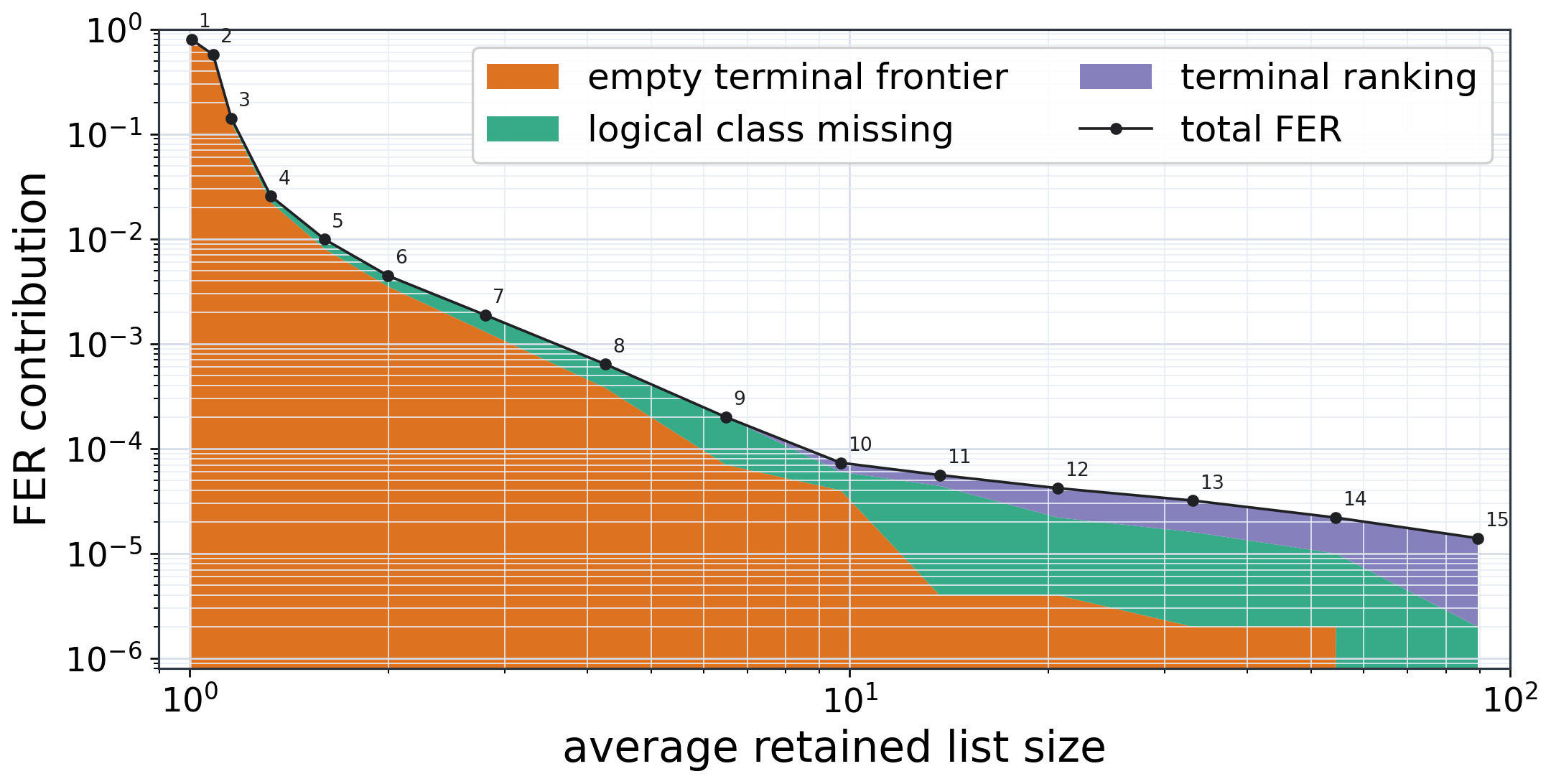}
\caption{Failure decomposition of \frontier{} for the gross code, under depolarizing circuit-level noise with $p=0.002$. The cap is fixed to $K=2^{14}$ for $\Delta \leq 12$ and $K=2^{14}$ for $\Delta =15$. 
Failures are classified into ``no path'' corresponding to shots where the terminal frontier is empty (no error with the correct syndrome was found); ``truth missing'' when the terminal list does not contain the true error; and ``bad ranking'' when the true error is in the terminal list, but not ranked first.}
\label{fig:failure_decomposition}
\end{figure}

Figure~\ref{fig:failure_decomposition} shows the failure decomposition for the gross code under circuit-level noise. 
There are three types of failures. The first one corresponds to an empty terminal frontier: this means that pruning was too aggressive and that all remaining states have the wrong final syndrome and are therefore removed from the terminal frontier. The second type corresponds to cases where the terminal frontier is not empty but does not contain the correct logical class. The final type of failure is when the correct logical class is present at the end, but not ranked first by the decoder, meaning that there is another class with a larger terminal mass. 
We find that increasing $\Delta$ (and the cap value $K$) quickly removes the first two types of failures. In particular, in practice, it is useful to set the decoder settings so that most remaining failures are bad terminal ranking. This typically leads to a good tradeoff between accuracy and complexity of the decoder.

\FloatBarrier

\section{Open problems}
\label{sec:problems}

We now list a number of possible avenues to improve the \frontier{} decoder further. 
The first one is to find a better column ordering, and maybe to rely on larger committees. For instance, one could try to take into account the weight or the prior of each column, or exploit local structure in the parity-check matrix. 
Note that for online decoders, which are of interest in order to go beyond simple memory experiments, it is very useful that the final ordering remains close to the temporal circuit ordering to avoid latency issues. This is the case of the deadline ordering for the codes we have tested. 

A second problem is to improve the current suffix-compatibility score. As explained, this score is a heuristic that completely ignores correlations between the rows of the parity-check matrix. We currently try to correct for this by adding a normalization factor $\alpha <1$, but it should be possible to do much better by taking detector correlations into account. In simulations, we have seen that considering pairwise correlations could improve the accuracy a little bit, but the tradeoff between complexity and accuracy led us to keep our simple heuristic.

Terminal ranking remains a bottleneck. Some failures are not caused by losing the right class during pruning; they are caused by ranking the terminal classes incorrectly. Better terminal aggregation, better calibration between prefix mass and best-path evidence, or a second-stage ranker on the retained terminal list may improve this step.
We note that when the final list contains several candidates, it is possible to compute a posterior gap. This information can be used either to accept the candidate or to decide that the result should not be trusted, and for instance rerun the decoder with larger values of $\Delta$ and $K$. 

Even before aggressive speed optimization, \frontier{} can serve as a
high-accuracy offline decoder for benchmarking faster decoders, identifying hard
syndromes, and diagnosing whether failures arise from support loss or terminal
ranking. The present work focuses on memory-style code-capacity and
detector-side DEM benchmarks. Extending the same ordered detector-picture
approach to logical circuits, calibrated hardware noise, and non-Pauli errors
is left for future work.

\section{Conclusion}

The \frontier{} decoder is based on a simple principle: order the variables,
keep the active boundary, merge equivalent prefixes, and prune by score gap,
using a hard cap only as a safeguard. Our simulations suggest
that this principle captures useful posterior structure in several quantum
decoding problems, notably for the surface code and for the color code in the code-capacity model, and for the surface code and BB codes under circuit-level noise.
In the latter case, \frontier{} is competitive with strong
heuristic decoders while using very compact retained lists.

The present results should be read as an initial algorithmic study, not as a
final optimized decoder. The main lesson from our experiments is that the useful object is
not a small set of individual errors, but a small set of boundary states.
This distinction is especially important for quantum decoding: many physical
errors may represent the same residual-syndrome/logical boundary state, and
\frontier{} sums their masses before pruning rather than forcing them to
compete as separate list entries. Degeneracy is therefore not only an ambiguity
to survive, but also a source of compression. The next steps are to improve the
boundary score, improve the terminal logical-class ranking, and understand
when a small ordered frontier captures most of the posterior mass needed for
the logical decision. This open problem for \frontier{} is
analogous to the bond-dimension question for tensor-network decoders, but with
a matrix-order boundary in place of a geometric cut: when do the parity-check
matrix, the noise model, and the ordering make the logical posterior
effectively low-width after degeneracy-induced merging?

\section*{Code availability}

The implementation of the frontier decoder used in this work is available at
\url{https://github.com/aleverrier/frontier}.

\begin{acknowledgments}
AL acknowledges the Plan France 2030 through the project ANR-22-PETQ-0006.
RU gratefully acknowledges the hospitality of the COSMIQ group at Inria, where this work was carried out during his sabbatical.

We acknowledge the use of large language models, in particular OpenAI Codex, to assist with the programming, documentation, and testing of the frontier decoder implementation.
\end{acknowledgments}

\bibliographystyle{quantumlinks}
\bibliography{frontier_decoder_refs}

\end{document}